\documentclass[sigconf]{acmart}
\usepackage{amsmath,amsthm,amssymb}
\usepackage{mathrsfs}
\usepackage{bbm}
\usepackage{enumerate}
\usepackage{graphicx}
\usepackage{url}
\usepackage{subcaption}
\usepackage{balance}

\usepackage{algpseudocode}
\usepackage{algorithmicx, algorithm}
\usepackage{color}

\newtheorem{theorem}{Theorem}
\newtheorem{lemma}[theorem]{Lemma}
\newtheorem{proposition}[theorem]{Proposition}
\newtheorem{claim}[theorem]{Claim}

\newtheorem{fact}[theorem]{Fact}

\newcommand{\set}[1]{\{#1\}}
\newcommand{\abs}[1]{\left\lvert#1\right\rvert}
\newcommand{\size}[1]{|#1|}

\newcommand{\floor}[1]{\lfloor#1\rfloor}
\newcommand{\ceil}[1]{\lceil#1\rceil}

\let\adj\sim

\let\emptyset\varnothing

\let\sss\scriptscriptstyle
\newcommand{\xmin}{x_{\min}}
\newcommand{\starnode}{v^*}

\newcommand{\ER}{Erd\"{o}s-R\'enyi}

\renewcommand{\epsilon}{\varepsilon}
\newcommand{\plexp}{\beta}
\newcommand{\weight}[1]{p_{#1}}

\newcommand{\w}[1]{p_{#1}}
\newcommand{\wb}[1]{p(#1)}
\newcommand{\seqdeg}{\mathbf{p}}
\newcommand{\avgdeg}{\nu}
\newcommand{\secdeg}{\omega}
\newcommand{\Exp}[1]{\mathbb{E}[#1]}
\newcommand{\E}{\ensuremath{\mathop{\mathbb{E}}}}
\newcommand{\vol}[1]{\mathrm{vol}(#1)}

\newcommand{\vols}[1]{\mathrm{vol}_2(#1)}
\newcommand{\level}{\Gamma}
\newcommand{\neigh}{N}
\newcommand{\pr}[1]{\mathrm{Pr}[#1]}

\newcommand{\maxdeg}{\frac 1 {\plexp - 1}}
\newcommand{\istrue}[1]{\mathbbm{1}_{#1}}

\newcommand{\dist}{\mathrm{dist}}
\newcommand{\argmax}{\mathrm{argmax}}

\newcommand{\maxdegWt}{\varepsilon(n)}
\newcommand{\labelExpo}{\min(\frac 1 {\plexp - 1}, \frac 1 {4 - \plexp})}
\newcommand{\labelSize}{1 - \labelExpo}

\newcommand{\dhat}[1]{\hat{d}_{#1}}
\newcommand{\threshold}{\frac 1 {(4 - \plexp)(\plexp - 1)}}

\algnewcommand\algorithmicinput{\textbf{Input:}}
\algnewcommand\Input{\item[\algorithmicinput]}
\algnewcommand\algorithmicoutput{\textbf{Output:}}
\algnewcommand\Output{\item[\algorithmicoutput]}

\newcommand{\cE}{\mathcal{E}}

\newcommand{\cN}{\mathcal{N}}

\newcommand{\cL}{\mathcal{L}}

\newcommand{\cO}{\mathcal{O}}

\newcommand{\volg}[1]{\mathrm{vol}_{2+\gamma}(#1)}

\newcommand{\Bracket}[1]{\left(#1\right)}

\newcommand{\polylog}[1]{\mathrm{poly}\log(#1)}

\newtheoremstyle{restate}{}{}{\itshape}{}{\bfseries}{~(restated).}{.5em}{\thmnote{#3}}
\theoremstyle{restate}

\setcopyright{none}

\acmConference[WWW]{The Web Conference}{2019}{San Francisco, CA}
\acmYear{2019}
\copyrightyear{2019}

\title{Pruning based Distance Sketches with Provable Guarantees on Random Graphs}

\author{Hongyang Zhang}
\affiliation{Stanford University}
\email{hongyang@cs.stanford.edu}

\author{Huacheng Yu}
\affiliation{Harvard University}
\email{yuhch123@gmail.com}

\author{Ashish Goel}
\affiliation{Stanford University}
\email{ashishg@stanford.edu}

\begin{document}

\begin{abstract}
	Measuring the distances between vertices on graphs is one of the most fundamental components in network analysis.
	Since finding shortest paths requires traversing the graph, it is challenging to obtain distance information on large graphs very quickly.
	In this work, we present a preprocessing algorithm that is able to create landmark based distance sketches efficiently, with strong theoretical guarantees.
	When evaluated on a diverse set of social and information networks, our algorithm significantly improves over existing approaches by reducing the number of landmarks stored, preprocessing time, or stretch of the estimated distances.

	On Erdos-Renyi graphs and random power law graphs with degree distribution exponent $2 < \beta < 3$, our algorithm outputs an exact distance data structure with space between $\Theta(n^{5/4})$ and $\Theta(n^{3/2})$ depending on the value of $\beta$, where $n$ is the number of vertices.
	We complement the algorithm with tight lower bounds for Erdos-Renyi graphs and the case when $\beta$ is close to two.

\end{abstract}

\maketitle

\section{Introduction}

Computing shortest path distances on large graphs is a fundamental
problem in computer science and has been the subject of much
study~\cite{TZ05, GH05, S14}. In many applications, it is important to compute the
shortest path distance between two given nodes, i.e. to answer
shortest path queries, in real time. Graph distances measure the closeness or similarity of vertices and are often used as one of the most basic metric in network analysis \cite{HWYY07,PBCG09,VFD07,YB08}. In this paper, we will focus on
efficient and practical implementations of shortest path queries in
classes of graphs that are relevant to web search, social networks,
and collaboration networks etc.  For such graphs, one commonly used
technique is that of {\em landmark-based labelings}: every node is
assigned a set of landmarks, and the distance between two nodes in
computed only via their common landmarks. If the set of landmarks can be easily
computed, and is small, then we obtain both efficient pre-processing
and small query time.

Landmark based labelings (and their more general counterpart, {\em
  Distance Labelings}), have been studied
extensively~\cite{S14,BDG16}. In particular, a sequence of
results culminating in the work of Thorup and Zwick~\cite{TZ05}
showed that labeling schemes can provide a multiplicative
3-approximation to the shortest path distance between any two nodes,
while having an overhead of ${O}({\sqrt n})$ storage per node on average in
the graph (we use the standard notation that a graph has $n$ nodes and
$m$ edges). In the worst case, there is no distance labeling scheme
that always uses sub-quadratic amount of space and provides exact
shortest paths.  Even for graphs with maximum degree $3$, it is known
that any distance labeling scheme requires $\Omega(n^{3/2})$
total space \cite{GPPR01}.  In sharp contrast to these theoretical
results, there is ample empirical evidence that very efficient
distance labeling schemes exist in real world graphs that can achieve
much better approximations. For example, Akiba et al.~\cite{AIY13} and
Delling et al.~\cite{DGPW14} show that current algorithms can find
{\em landmark based labelings} that use only a few hundred landmarks
per vertex to obtain {\em exact distance}, in a wide collection of
social, Web, and computer networks with millions of vertices. In this
paper, we make substantial progress towards closing the gap between
theoretical and observed performance. We show that natural landmark
based labeling schemes can give exact shortest path distances with a
small number of landmarks for a popular model of (unweighted and
undirected) web and social graphs, namely the heavy-tailed random
graph model. We also formally show how further reduction in the number
of landmarks can be obtained if we are willing to tolerate an additive
error of one or two hops, in contrast to the multiplicative
3-approximation for general graphs. Finally, we present practical
versions of our algorithms that result in substantial performance
improvements on many real-world graphs.

In addition to being simple to implement, landmark based shortest path
algorithms also offer a qualitative benefit, in that they can directly
be used as the basis of a social search algorithm. In social
search~\cite{BG12}, we assume there is a collection of keywords
associated with every node, and we need to answer queries of the
following form: given node $v$ and keyword $w$, find the node that is
closest to $v$ among all nodes that have the keyword $w$ associated
with them. This requires an index size that is $O(L)$ times the size
of the total social search corpus and a query time of $O(L)$, where $L$
is the number of landmarks per node in the underlying landmark based
algorithm; the approximation guarantee for the social search problem
is the same as that of the landmark based algorithm. Thus, our results
lead to both provable and practical improvements to the social search
problem.


Existing models for social and information networks build on random
graphs with some specified degree distribution \cite{D07,CL06,V09},
and there is considerable evidence that real-world graphs have
power-law degree distributions~\cite{CSN09,EG17}. We will use the
Chung-Lu model~\cite{CL02}, which assumes that the degree sequence of
our graph is given, and then draws a ``uniform'' sample from graphs
that have the same or very similar degree sequences. In particular, we
will study the following question:
{\em Given a random graph from the Chung-Lu model with a power law
  degree distribution of exponent $\beta$, how much storage does a
  landmark-based scheme require overall, in order to answer distance
  queries with no distortion?}

In the rest of the paper, we use the term ``random power law graph''
to refer to a graph that is sampled from the Chung-Lu model, where the
weight (equivalently, the expected degree) of each vertex is
independently drawn from a power law distribution with exponent
$\beta$.  We are interested in the regime when $\plexp > 2$
{\textemdash} this covers most of the empirical power law degree
distributions that people have observed on social and information
networks~\cite{CSN09}. Admittedly, real-world graphs have additional
structure in addition to having power-law degree
distributions~\cite{LLDM08}, and hence, we have also validated the effectiveness of our algorithm on real graphs.

\subsection{Our Results}

Our first result corresponds to the ``easy regime'', where the degree
distribution has finite variance ($\beta > 3$). We show that a simple
procedure for generating landmarks guarantees exact shortest paths,
while only requiring each node to store $\tilde{O}(\sqrt n)$
landmarks. The same conclusion also applies to Erd\H{o}s-Renyi graphs
$G(n, \frac c n)$ when $c > 1$, or when $c = 2\log n$.

We then study the case where $2 < \plexp < 3$. This is the most
emblematic regime for power-law graphs, since the degree distribution
has infinite variance but finite expectation. We present an algorithm
that generates at most $\tilde{O}(n^{\sss (\beta - 2)/(\beta - 1)})$
landmarks per node when $\beta \ge 2.5$; and
$\tilde{O}(n^{\sss (3-\beta)/(4 - \beta)})$ landmarks per node when
$2 < \beta < 2.5$. We obtain additional improvements if we allow an
additive error of 1 or 2. See Figure \ref{fig:result} for an
illustration of our results.

\vspace{-0.075in}
\begin{figure}[thbp]
	\centering
	\includegraphics[width=2.4in]{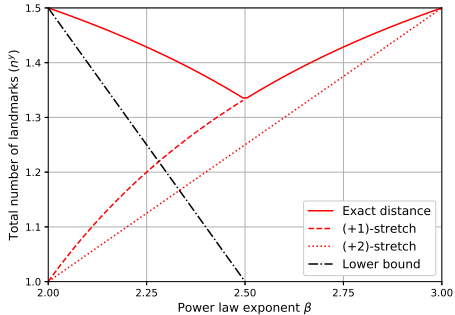}
	\vspace{-0.125in}
	\caption{An illustration of the results for labeling schemes:
	The $x$-axis is the exponent of the power law degree distribution and each value on the $y$-axis corresponds to a storage of $\tilde{\textsc O}(n^{y})$.
	The lower bound curve is for exact distances.}
	\label{fig:result}
\end{figure}

\vspace{-0.125in}
While the dependence on $\beta$ is complex, it is worth noting that in
the entire range that we study ($\beta > 2$), the number of landmarks
per node is at most $\tilde{O}({\sqrt n})$ for {\em exact shortest
  paths}.  This is in stark contrast to known impossibility results
for general graphs, where no distance labeling with a multiplicative
stretch less than 3 can use sub-linear space per
node~\cite{GPPR01}. The query time of our algorithms is proportional to the number
of landmarks per node, so we also get speed improvements.

Our algorithm is based on the pruned labeling algorithm of Akiba et
al. \cite{AIY13}, but differs in important ways.  The pruned labeling
algorithm initially posits that every node is a landmark for every
other node, and then uses the BFS tree from each node to iteratively
prune away unnecessary landmarks.  In our approach, we apply a similar BFS with pruning procedure on a
small subset of $H$ (i.e. high degree vertices), but switch to
lightweight local ball growing procedures up to radius $l$ for all
other vertices.  As we show, the original pruned labeling algorithm
requires storing $\tilde{\Omega}(n^2)$ landmarks on sparse \ER~graphs.
By growing local balls of size $\sqrt n$, our algorithm recovers exact
distances with at most $\tilde{O}(n^{3/2})$ landmarks instead, for
\ER~graphs and random power law graphs with $\beta > 3$. Hence, our
algorithm combines exploiting the existence of high-degree ``central
landmarks'' with finding landmarks that are ``locally important''.
Furthermore for $2 < \beta < 3$, by setting up the number of global landmarks $H$ and the
radius $l$ suitably, we provably recover the upper bounds described in
Figure \ref{fig:result}. While the algorithmic idea is quite simple,
the analysis is intricate.

We complement our algorithmic results with tight lower bounds for the
regime when $\beta > 3$: the total length of any distance labeling
schemes that answer distance queries exactly in this regime is almost
surely $\tilde{\Omega}(n^{1.5})$. We also show that when $2 < \beta < 2.5$,
any distance labeling scheme will generate labels of total size
$\tilde{\Omega}(n^{3.5-\plexp})$ almost surely.  In particular, our algorithm
achieves the optimal bound when $\beta$ is close 2.


The parameter choice suggested by our theoretical analysis can be quite expensive to implement (as can earlier landmark based algorithms).
We apply a simple but principled parameter tuning procedure to
our algorithm that substantially improves the preprocessing time and
generates a smaller set of landmarks at essentially no loss of
accuracy. We conduct experiments on several real world graphs, both
directed and undirected.  First, compared to the pruned labeling
algorithm, we find that our algorithm reduces the number of landmarks
stored by 1.5-2.5x; the preprocessing time is reduced significantly as
well.  Next, we compare our algorithm to the distance
oracle of Thorup and Zwick~\cite{TZ05}, which is believed to be theoretically optimal
for worst-case graphs, as well as the distance sketch of Das Sarma et
al~\cite{DGNP10} which has been found to be both efficient and useful
in prior work~\cite{BG12}. For each graph, our algorithm substantially
outperforms these two benchmarks.  Details are in
Section~\ref{sec:experiment}. It is important to note that the three
algorithms we compare to also work much better on these real-world
graphs than their theoretical guarantee, and we spend considerable
effort tuning their parameters as well. Hence, the performance
improvement given by our algorithm is particularly noteworthy.

It is worth mentioning that our technical tools only rely on bounding
the growth rate of the breadth-first search. Hence we expect that our
results can be extended to the related configuration model~\cite{D07}
as well.
One limitation of our work is that the analysis does not apply directly to preferential attachment graphs, which correspond to another family of well known power law graphs.
But we believe that similar results can be obtained there by adapting our analysis to that setting as well.
This is left for future work.

\smallskip
\noindent{\bf Organizations:} The rest of the paper is organized as follows.
Section \ref{sec:prelim} introduces the basics of random graphs, reviews the pruned labeling algorithm and related work.
Section \ref{sec_alg} introduces our approach.
Section \ref{sec_powerlaw} analyzes our algorithm on random power law graphs.
Then we present experiments in Section \ref{sec:experiment}.
We show the lower bounds in Section \ref{sec_lb}.
We conclude in Section \ref{sec:discuss}.
The Appendix contains missing proofs from the main body.
\section{Preliminaries and Related Work}\label{sec:prelim}
\newcommand{\algPrune}{{\sc PrunedLabeling}}
\newcommand{\bfsFw}{{\sc forwardBfs}}
\newcommand{\bfsBk}{{\sc backwardBfs}}

\subsection{Notations}
Let $G = (V, E)$ be a directed graph with $n = \abs{V}$ vertices and $m = \abs{E}$ edges.
For a vertex $x \in V$, denote by $d_{out}(x)$ the outdegree of $x$ and $d_{in}(x)$ the indegree of $x$.
Let $\cN_{out}(x)$ denote the set of its out neighbors.
Let $\dist_G(x,y)$ denote the distance of $x$ and $y$ in $G$,
or $\dist(x,y)$ for simplicity.
When $G$ is undirected, then the outdegrees and indegrees are equal.
Hence we simply denote by $d_x$ the degree of every vertex $x \in V$.
For an integer $l$ and $x \in V$, denote by $\Gamma_l(x) = \set{y: \dist(x, y) = l}$ the set of vertices at distance $l$ from $x$.
Denote by $N_l(x)$ the set of vertices at distance at most $l$ from $x$.

We use notation $a \lesssim b$ to indicate that there exists an absolute constant $C > 0$ such that $a \le C b$.
The notations $\tilde{O}(\cdot)$ and $\tilde{\Omega}(\cdot)$ hide poly-logarithmic factors.

\subsection{Landmark based Labelings}

In a landmark based labeling \cite{DGSW14}, each vertex $x$ is assigned a set of forward landmarks $L_F(x) $ and backward landmarks $L_B(x)$.
Each landmark set is a hash table, whose keys are vertices and values are distances.
For example, if $y \in L_F(x)$, then the value associated with $y$ would be $\dist(x, y)$, which is the ``forward'' distance from $x$ to $y$.
Given the landmark sets $L_F(\cdot)$ and $L_B(\cdot)$, we estimate the distances as follows:
\begin{align}\label{eq_query}
	\min_{z \in L_F(x) \cap L_B(y)} \dist(x, z) + \dist(z, y), \forall x, y \in V.
\end{align}
If no common vertex is found between $L_F(x)$ and $L_B(y)$,
then $y$ is not reachable from $x$.
In the worst case, computing set intersection takes $\Omega(\min(\abs{L_F(x)}, \abs{L_B(y)}))$ time.

Denote the output of equation \eqref{eq_query} by $\hat{d}$.
Clearly, we have $\hat{d} \ge \dist(x, y)$ when $y$ is reachable from $x$.
The {\it additive stretch} of $\hat{d}$ is given by $\hat{d} - \dist(x, y)$, and the {\it multiplicative stretch} is given by $\hat{d} / \dist(x, y)$.
When there are no errors for any pairs of vertices, such landmark sets are called {\it 2-hop covers} \cite{C97}.

There is a more general family of data structures known as labeling schemes \cite{GPPR01}, which associates a vector $\cL: V\rightarrow \{0,1\}^*$ for every vertex.
When answering a query for a pair of vertices $x, y \in V$, only the two labels $\cL(x)$ and $\cL(y)$ are required without accessing the graph.
The total length of $\cL$ is given by $\sum_{x \in V} \size {\cL(x)}$.
It is clear from equation \eqref{eq_query} that landmark sketches fall in the query model of labeling schemes.

\subsection{The Pruned Labeling Algorithm}\label{sec_prune}

We review the pruned labeling algorithm \cite{AIY13} for readers who are not familiar.
The algorithm starts with an ordering of the vertices, $\set{x_1, x_2, \dots, x_n}$.
First for $x_1$, a breadth first search (BFS) is performed over the entire graph.
During the BFS, $x_1$ is added to the landmark set of every vertex.%
\footnote{For directed graphs, there will be a forward BFS which looks at $x_1$'s outgoing edges and its descendants, as well as a backward BFS which looks at $x_1$'s incoming edges and its predecessors.}
Next for $x_2$, in addition to running BFS, a {\it pruning} step is performed before adding $x_2$ as a landmark.
For example, suppose that a path of length $l$ is found from $x_2$ to $y$.
If $x_1$ lies on the shortest path fom $x_2$ to $y$, then by checking their landmark sets, we can find the common landmark $x_1$ to certify that $\dist(x_2, y) = \dist(x_2, x_1) + \dist(x_1, y) \le l$.
In this case, $x_2$ is not added to $y$'s landmark set, and the neighbors of $y$ are pruned away.
The above procedure is repeated on $x_3$, $x_4$, etc.

For completeness, we describe the pseudocode in Algorithm \ref{alg_prune}.
Note that the backward BFS procedure can be derived similarly.
It has been shown that the pruned labeling algorithm is guaranteed to return exact distances \cite{AIY13}.%

\begin{algorithm}[!h]
	\small
	\caption{\algPrune~(Akiba et al. \cite{AIY13})}
	\label{alg_prune}
	\begin{algorithmic}[1]
		\Input A directed graph $G = (V, E)$;
		An ordering of $V$, $\set{x_1, x_2, \dots, x_n}$.
		\State Let $\cO = \varnothing$, and $L_{F}(x) = L_{B}(x) = \varnothing$, for all $x \in V$
		\For {$i = 1, \dots, n$}
			\State \bfsFw($x_i$)
			\State \bfsBk($x_i$)
			\State $\cO = \cO \cup \set{x_i}$
		\EndFor
		\\
		\Procedure{\bfsFw}{$x_i$}
			\State Let $Q$ be a priority queue and $S$ be a hash set
			\State Set the priority of $x_i$ to be zero
			\While {$Q \neq \varnothing$}
				\State Let $l$ be the minimum priority of $Q$
				\State Let $u$ be the corresponding vertex
				\State $S = S \cup \set{u}$
				\State Let $\tilde{d} = \min_{y \in L_F(x_i) \cap L_B(u)} \dist(x_i, y) + \dist(y, u)$ 
				\If {$l < \tilde d$} \hfill{(otherwise $u$'s neighbors are pruned)}
					\State $L_B(u) = L_B(u) \cup \set{x_i \rightarrow l}$ 
					\For {$y \in \cN_{out}(u)$ such that $y \notin \cO$}
						\State Let $q$ be the priority of $y$
						\If {$y \notin S$ and $l+1 < q$}
							\State Decrease $y$'s priority to $l+1$
						\EndIf
					\EndFor
				\EndIf
			\EndWhile
		\EndProcedure
	\end{algorithmic}
\end{algorithm}%
\vspace{-0.1in}

\subsection{Random Graphs}

We review the basics of \ER~random graphs.
Let $G = G(n, p)$ be an undirected graph where every edge appears with probability $p$.
It is well known that when $p \ge 2 (\log n) / n$, $G$ has only one connected component with high probability.
Moreover, the neighborhood growth rate (i.e. $\abs{\Gamma_{i+1}(x)} / \abs{\Gamma_i(x)}$) is highly concentrated around its expectation, which is $np$.
Formally, the following facts are well-known.
\begin{fact}[Bollob{\'a}s \cite{B98}]\label{fact_er}
	Let $G = (V, E)$ be an undirected graph where every edge is sampled with probability $p = 2(\log n) / n$.
	Let $D = \ceil{\frac{\log n} {\log (np)}}$.
	Then the following are true with high probability:
	\begin{itemize}
		\item[a)] The diameter of $G$ is at most $D+1$;
		\item[b)] For any $x, y \in V$ and $l \le D$, $\Pr(\dist(x, y) \le l) \le \frac{(np)^{l+1}} {n(np - 1)}$;
		\item[c)] For any $x\in V$ and $l < D$, we have $\frac 1 2 \le \frac{\abs{\Gamma_l(x)}} {(np)^l} \le 2$.
	\end{itemize}
\end{fact}

\noindent{\bf The Chung-Lu model:} Let $p_x > 0$ denote a weight value for every vertex $x \in V$.
Given the weight vector $\mathbf{p}$, the Chung-Lu model generalizes \ER~graphs such that each edge is chosen independently with probability
\[ \Pr[x \sim y] = \min \left\{\frac {\w x \cdot \w y} {\vol{V}}, 1\right\}, \forall x, y \in V \]
where $\vol{V} = \sum_{x \in V} p_x$ denote the volume of $V$.

\smallskip
\noindent{\bf Random power law graphs:}
Let $f: [\xmin, \infty) \rightarrow \mathbb R$
denote the probability density function of a power law distribution
with exponent $\plexp > 1$, i.e. $f(x) = Z x^{-\plexp}$,
where $Z = (\plexp - 1) \cdot \xmin^{\plexp - 1}$ \cite{CSN09}.
The expectation of $f(\cdot)$ exists when $\beta>2$.
The second moment is finite when $\beta>3$. 
When $\beta < 3$, the expectation is finite, but the empirical second moment grows polynomially in the number of samples with high probability.
If $\beta < 2$, then even the expectation becomes unbounded as $n$ grows.

In a random power law graph, the weight of each vertex is drawn independently from the power law distribution.
Given the weight vector $\seqdeg$, a random graph is sampled
according to the Chung-Lu model.
If the average degree $\avgdeg > 1$, then it is known that almost surely the graph has a unique giant component \cite{CL06}.

\subsection{Related Work}

\noindent{\bf Landmark based labelings:}
There is a rich history of study on how to preprocess a graph to answer shortest path queries \cite{AG11, ADKM15, ADDJ93,C97,TZ05,BCT17}.
It is beyond our scope to give a comprehensive review of the literature and we refer the reader to survey \cite{S14} for references.

In general, it is NP-hard to compute the optimal landmark based labeling (or 2-hop cover).
Based on an LP relaxation, a $\log n$ factor approximation can be obtained via a greedy algorithm \cite{CHKZ03}.
See also the references \cite{AMO17, BGKSW15, DGSW14, GRS13, B17} for a line of followup work.
The current state of the art is achieved based on the pruned labeling algorithm \cite{AIY13, DGPW14}.
Apart from the basic version that we have already presented, bit-parallel optimizations have been used to speed up proprocessing \cite{AIY13}.
Variants which can be executed when the graph does not fit in memory have also been studied \cite{J14}.
It is conceivable that such techniques can be added on top of the algorithms that we study as well.
For the purpose of this work, we will focus on the basic version of the pruned labeling algorithm.
Compared to classical approaches such as distance oracles, the novelty of the pruned labeling algorithm is using the idea of \textit{pruning} to reduce redundancy in the BFS tree.

\smallskip
\noindent{\bf Network models:}
Earlier work on random graphs focus on modeling the small world phenomenon \cite{CL06}, and show that the average distance grows logarithmically in the number of vertices.
Recent work have enriched random graph models with more realistic features, e.g. community structures \cite{KPPS14}, shrinking diameters in temporal graphs \cite{LCKF10}.

Other existing mathematical models on special families of graphs related to distance queries include road networks \cite{AFGW10}, planar graphs and graphs with doubling dimension.
However none of them can capture the expansion properties that have been observed on sub-networks
of real-world social networks \cite{LLDM08}.

Previous work of Chen et al. \cite{CSTW09} presented a 3-approximate labeling scheme requiring storage $\tilde{O}(n^{(\plexp-2) /(2\plexp-3)})$ per vertex, on random power law graphs with $2 < \beta < 3$.
Our (+2)-stretch result improves upon this scheme in the amount of storage needed per vertex for $2 < \beta < 2.5$, with a much better stretch guarantee.
Another related line of work considers compact routing schemes on random graphs.
Enachescu et al.~\cite{EWG08} presented a 2-approximate compact routing scheme using space $O(n^{\sss 1.75})$ on \ER{} random graphs,
and Gavoille et al. \cite{GGHI15} obtained a 5-approximate compact routing scheme on random power law graphs.


\section{Our Approach}\label{sec_alg}

In order to motivate the idea behind our approach, we begin with an analysis of the pruned labeling algorithm on \ER~random graphs.
While the structures of real world graphs are far from \ER~graphs, the intuition obtained from the analysis will be useful.
Below we describe a simple proposition which states that for sparse \ER~graphs, the pruned labeling algorithm outputs $\tilde{\Omega}(n^2)$ landmarks.

\begin{proposition}\label{prop_prune_ER}
	Let $G = (V, E)$ be an undirected \ER~graph where every edge appears with probability $p = 2(\log n) / n$.
	For any ordering of the vertices $V = \set{x_1, x_2, \dots, x_n}$, with high probability over the randomness of $G$, the total number of landmarks produced by Algorithm \ref{alg_prune} is at least $\tilde{\Omega}(n^2)$.
\end{proposition}

\begin{proof}[Proof sketch]
	We first introduce a few notations.
	Let $r=np$ denote the growth rate of $G$.
	Let $\epsilon=1/\log n$.
	Consider a vertex $x_i$ where $1 \le i \leq \epsilon n$.
	Denote by $X_{-i} = \{x_1,\ldots,x_{i-1}\}$.
	Consider any $u \in V$, if none of the shortest paths from $x_i$ to $u$ intersect with $X_{-i}$, then $(x_i, u)$ is called a {\it bad} pair.
	Note that $x_i$ must be added to $u$'s landmark set by Algorithm \ref{alg_prune}, because during the BFS from $x_i$, all estimates through $X_{-i}$ will be strictly larger than $\dist(x_i, u)$.
	Hence, to lower bound the total landmark sizes, it suffices to count the number of {\it bad} pairs.
	In the following, we show that in expectation for every $x_i$ where $1 \le i \leq\epsilon n$, there are at least $n/(\log n)^3$ vertices $u$ such that $(x_i, u)$ are {\it bad}.
	It follows that Algorithm \ref{alg_prune} requires at least $\epsilon n^2/(\log n)^3\geq \tilde{\Omega}(n^2)$ in expectation.

	Let $D = \lfloor\log_r n-2\rfloor$.
	Consider $\Gamma_D(x_i)$, the set of vertices at distance equal to $D$ from $x_i$.
	We count the number of {\it bad} vertices in $\Gamma_D(x_i)$ at follows.
	For each $1 \le k \le D$, consider the intersection $\Gamma_k(x_i) \cap X_{-i}$ and their subtree down to $\Gamma_D(x_i)$.

	Starting from any $y \in \Gamma_k(x_i) \cap X_{-i}$, the subtree of $y$ would result in {\it good} vertices in $\Gamma_D(x_i)$, whose distance from $x_i$ can be correctly estimated (c.f. line 15-16 in Algorithm \ref{alg_prune}).
	In expectation, the size of the intersection is $r^k \varepsilon$, because the probability that any two vertex has distance $k$ on $G$ is equal to $r^k / n$, and there are at most $\varepsilon n$ vertices in $X_{-i}$.
	Next, each $y$ results in $r^{D- k}$ vertices in its $(D-k)$-th level neighborhood.
	Combined together, the total number of {\it good} vertices which are covered by $\Gamma_k(x_{i}) \cap X_{-i}$ is at most
	$r^k \varepsilon \times r^{D-k} = \varepsilon r^D$.
	By summing over all $k \le D$, we obtain that the total number of {\it good} vertices in $\Gamma_D(x_i)$ is at most $D \varepsilon r^D$.

	On the other hand, the size of $\Gamma_D(x_i)$ is $r^D$.
	Hence the total number of {\it bad} vertices is at least $(1 - D\varepsilon) r^D \ge n / \log^3 n$.
	To show that the the proposition holds with high probability, it suffices to apply concentration results on neighborhood growth in the arguments above.
	We omit the details.

\end{proof}

The interesting point from the above analysis is that $\Theta(n)$ landmarks are added throughout the first $\varepsilon n$ vertices.
The reason is that there are no high degree vertices in \ER~graphs, hence the landmarks we have added in the beginning do not cover the shortest paths for many vertex pairs later.%
Secondly, a large number of distant vertices are added in the landmark sets, which do not lie on the shortest paths of many pairs of vertices.

Motivated by the observation, we introduce our approach as follows.
We start with an ordering of the vertices.
For the top $H$ vertices in the ordering, we perform the same BFS procedure with pruning.
For the rest of the vertices, we simply grow a local ball up to a desired radius.
Concretely, only the vertices from the local ball will be used as a landmark.
Algorithm \ref{alg_approx} describes our approach in full.\footnote{Here we have omitted the details of the local backward BFS procedure, which can be derived similar to the local forward BFS procedure.}
As a remark, when the input graph $G$ is undirected, it suffices to run one of the forward or backward BFS procedures, and for each vertex, its forward and backward landmark sets can be combined to a single landmark set.

\newcommand{\algApprox}{{\sc ApproximatePruning}}
\newcommand{\localBfsFw}{{\sc localForwardBfs}}
\newcommand{\localBfsBk}{{\sc localBackwardBfs}}
\begin{algorithm}[!t]
	\small
	\caption{\algApprox}
	\label{alg_approx}
	\begin{algorithmic}[1]
		\Input A directed graph $G = (V, E)$;
		An ordering of $V$ $\set{x_1, x_2, \dots, x_n}$;
		The number of global landmarks $H$;
		The set of radiuses $\set{l_i}_{i=H+1}^n$.
		\State Let $\cO = \varnothing$, and $L_F(x) = L_B(x) = \varnothing$, for any $x \in V$.
		\For {$i = 1, \dots, H$}
			\State \bfsFw($x_i$)
			\State \bfsBk($x_i$)
			\State $\cO = \cO \cup \set{x_i}$
		\EndFor
		\For {$i = H+1, \dots, n$}
			\State \localBfsFw($x_i, l_i$)
			\State \localBfsBk($x_i, l_i$)
		\EndFor

		\Procedure{\localBfsFw}{$x_i, l_i$}
			\For {$y$ such that $\dist(x_i, y) \le l_{i} - 1$}
				\State $L_F(x_i) = L_F(x_i) \cup (y \rightarrow \dist(x_i, y))$
			\EndFor
			\For {$y$ such that $\dist(x_i, y) = l_i$ and $\exists z$ s.t. $\dist(x, z) = l_i - 1, (z, y) \in E, d_{out}(z) \leq d_{out}(y)$}
				\State $L_F(x_i) = L_F(x_i) \cup (y \rightarrow \dist(x_i, y))$ 
			\EndFor
		\EndProcedure
	\end{algorithmic}
\end{algorithm}

Recall that the backward and forward BFS procedures add a pruning step before enqueing a vertex (c.f. Algorithm \ref{alg_prune}).
 For each $x_i$ with $i > H$, the parameter $l_i$ controls the depth of the local ball we grow from $x_i$.
Furthermore, at the bottom layer, we only add vertices whose outdegree is higher than any of its predecessor to $x_i$'s landmark set.
The intuition is that vertices with higher outdegrees are more likely to be used as landmarks.


We begin by analyzing Algorithm \ref{alg_approx} for \ER~graphs, as a comparison to Proposition \ref{prop_prune_ER}.
The following proposition shows that without using global landmarks, local balls of suitable radius suffice to cover all the desired distances.
The proof is by observing that for each vertex, if
we add the closest $\sqrt n$ vertices to the landmark set of every vertex, then the landmark sets of every pair of vertices will intersect with high probability, i.e. we have obtained a 2-hop cover.

\begin{proposition}\label{prop_approx}
	Let $G = (V, E)$ be an undirected random graph where each edge is sampled with probability $p = 2 (\log n) / n$.
	By setting $H = 0$ and $l_i = l = \ceil{\frac {\log n} {2 \log {np}}} + 1$ for all $1 \le i \le n$, we have that Algorithm \ref{alg_approx} outputs a 2-hop cover with at most $\tilde{O}(n^{3/2})$ landmarks with high probability.
\end{proposition}

\begin{proof}
	Denote by $L(x)$ the landmark set obtained by Algorithm \ref{alg_approx}, for every $x \in V$.
	We will show that with high probability:
	\begin{itemize}
		\item[a)] For all $x_i, x_j \in V$, $L(x_i) \cap L(x_j) \neq \varnothing$. This implies that $L(\cdot)$ is a 2-hop cover.
		\item[b)] The size of $L(x_i)$ is less than $\tilde{O}(\sqrt n)$, for all $x_i \in V$.
	\end{itemize}
	Claim a) follows because the diameter of $G$ is at most $2l - 1$ with high probability by Fact \ref{fact_er}.
	Note that $L(x_i)$ contains $N_{l-1}(x_i)$, the set of vertices with distance at most $l - 1$. 
	If $\dist(x_i, x_j)\leq (l-1)+(l-1)$, $N_{l-1}(x_i)$ and $N_{l-1}(x_j)$ already intersect.
	Otherwise, since the diameter is at most $2l-1$, these two neighborhoods must be connected by an edge $e$.
	Suppose between $e$'s two endpoints, the one with a lower degree is on $x_i$'s side, then the local BFS from $x_i$ must add the other endpoint to $L(x_i)$, and vice versa.
	Therefore, $L(x_i)$ must intersect with $L(x_j)$.

	Claim b) is because $L(x_i)$ is a subset of $N_{l}(x_i)$. 
	By Fact \ref{fact_er}, the size of $N_{l}(x_i)$ is at most $4 (np)^{l} \lesssim \tilde{O}(\sqrt n)$.
	Hence, the proof is complete.
\end{proof}

\section{Random Power Law Graphs}\label{sec_powerlaw}

In this section we analyze our algorithm on random power law graphs.
We begin with the simple case of $\beta > 3$, which generalizes the result on \ER~graphs.
Because the technical intuition is the same with Proposition \ref{prop_approx}, we describe the result below and omit the proof.

\begin{proposition}\label{prop_beta3}
	Let $G = (V, E)$ be a random power law graph with average degree $\avgdeg > 1$ and power law exponent $\beta > 3$.
	For each $x_i \in V$, let $l_i$ be the smallest integer such that the number of edges between $N_{l_i}(x_i)$ and $V \backslash N_{l_i}(x_i)$ is at least $\delta \sqrt{n}$, where $\delta = 5 \sqrt{\avgdeg \log n}$.

	By setting $H = 0$ and $\set{l_i}_{i=1}^n$, Algorithm \ref{alg_approx} outputs a 2-hop cover with high probability.
	Moreover, each vertex uses at most ${O}(\sqrt n \log^2 n)$ landmarks.
\end{proposition}

\noindent{\bf Remark:} The high level intuition behind our algorithmic result is that as long as the breadth-first search process of the graph grows neither too fast nor too slow, but rather at a proper rate, then an efficient distance labeling scheme can be obtained.
Proposition \ref{prop_beta3} can be easily extended to configuration models with bounded degree variance.
It would be interesting to see if our results extend to
preferential attachment graphs and Kronecker graphs.

\bigskip
\noindent{\bf The case of $2 < \beta < 3$:} Next we describe the more interesting case with power law exponent $2 < \beta < 3$.
Here the graph contains a large number of high degree vertices.
By utilizing the high degree vertices, we show how to obtain exact distance landmark schemes, (+1)-stretch schemes and (+2)-stretch schemes.
The number of landmarks used varies depending on the value of $\beta$.
We now state our main result as follows.


\begin{theorem}\label{thm_sp2_ub}
	Let $G = (V, E)$ be a random power law graph with average degree
	$\avgdeg > 1$ and exponent $2< \beta < 3$.
	Let
	\begin{align*}
		K = \begin{cases}
				\sqrt n, \mbox{for }~2.5 \le \plexp \le 3 \\
				n^{\threshold}, \mbox{for }~2 < \plexp < 2.5.
		\end{cases}
	\end{align*}
	Let $H$ be the number of vertices whose degree is at least $K$ in $G$.
	Let $\pi = \set{x_i}_{i=1}^n$ be any ordering of vertices $V$ by their degrees in a non-increasing order.
	For each vertex $x_i \in V$, let $l_i$ be the smallest integer such that the number of edges between $N_{l_i - 1}(x_i)$ and $V \backslash N_{l_i -1}(x_i)$ is at least $\delta n^{(\plexp-2) / (\plexp-1)}$, where $\delta = 4 \avgdeg\cdot \log^2 n$.

	With ordering $\pi$, parameters $H$ and $\set{l_i}_{i=H+1}^n$,
	Algorithm \ref{alg_approx} outputs a 2-hop cover with high probability.
	Moreover, the maximum number of landmarks used by any vertex is at most
	\[ O\left(\max\left(n^{\frac{\beta - 2}{\beta - 1}}, n^{\frac{3 - \beta} {4 - \beta}}\right) \log^3 n\right). \]
\end{theorem}

The above theorem says that in Algorithm $\ref{alg_approx}$, first we use vertices whose degrees are at least $K$ as global landmarks.
Then for the other vertices $x_i$, we grow a local ball of radius $l_i$, whose size is (right) above $n^{ (\plexp-2) / (\plexp-1)}$.
The two steps together lead to a 2-hop cover.
We now build up the intuition for the proof.

\bigskip
\noindent{\bf Building up a $+1$-stretch scheme:}
First, it is not hard to show that $G$ contains a heavy vertex
whose degree is $n^{1/(\beta - 1)}$, by analyzing the power law distribution.
Note that $K \le n^{1 / (\beta - 1)}$, hence we have added all such high degree vertices as global landmarks.
This part, together with the local balls, already gives us a $(+1)$-stretch landmark scheme.

To see why, consider two vertices $x_i, x_j$.
If their local balls (of size $n^{(\beta-2) / (\beta-1)}$) already intersect, then we can already compute their distances correctly from their landmark sets.
Otherwise, since the bottom layers of $x_i$ and $x_j$ already have weight/degree $n^{(\beta - 2) / (\beta - 1)}$,
they are at most two hops apart, by connecting to the heavy vertex with degree $n^{1 / (\beta - 1)}$.
Recall that the heavy vertex is added to the landmark sets of every vertex.
Hence, the estimated distance is at most off by one.
As a remark, to get the $(+1)$-stretch landmark scheme, the number of landmarks needed per vertex is on the order of $n^{(\beta - 2) / (\beta - 1)}$.
This is because we only need to use vertices whose degrees are at least $n^{1 / (\beta - 1)}$ as global landmarks (there are only $\log n$ of them), as opposed to $H$ in Theorem \ref{thm_sp2_ub}.

\bigskip
\noindent{\bf Fixing the $+1$-stretch:} To obtain exact distances, for each vertex on the boundary of radius $l_i - 1$, we add all of its neighbors with a higher degree to the landmark set (c.f. line 15-17 in Algorithm \ref{alg_approx}).
Whenever there is an edge connecting the two boundaries, the side with a lower degree will add the other endpoint as a landmark, which resolves the (+1)-stretch issue.
For the size of landmark sets, it turns out that fixing the $(+1)$-stretch for the case $2 < \beta < 2.5$ significantly increases the number of landmarks needed.
Specifically, the costs are $n^{(5 - \beta) / (4 - \beta)}$ landmarks per node.

\bigskip
\noindent{\bf Intuition for the $+2$-stretch scheme:} As an additional remark, one can also obtain a $(+2)$-stretch landmark sketch by setting $l_i$ in Algorithm \ref{alg_approx} in a way such that every vertex stores the closest $\tilde{\Theta}(n^{(\beta / 2) - 1})$ vertices in its landmark set.
This modification leads to a $(+2)$-stretch scheme, because for two vertices $x, y$, once the bottom layers of $x, y$ have size at least $\tilde{\Theta}(n^{(\beta/2) -1})$, they are at most three hops away from each other.
The reason is that with high probability, the bottom layer will connect to a vertex with weight $\Omega(\sqrt n)$ in the next layer (which will all be connected), as it is not hard to verify that the volume of all vertices with weight $\sqrt n$ is $\Omega(n^{(4 - \beta) / 2})$.
By a similar proof to Theorem \ref{thm_sp2_ub}, the maximum number of landmarks used per vertex is at most $\tilde{O}(n^{(\beta - 2) / 2})$.

\bigskip
We refer the reader to Appendix \ref{sec:pf2} for details of the full proof.
The technical components involve carefully controlling the growth of the neighborhood sets by using concentration inequalities.

\section{Experiments}\label{sec:experiment}

In this section, we substantiate our results with experiments on a diverse collection of network datasets.
A summary of the findings are as follows.
We first compare Algorithm \ref{alg_approx} with the pruned labeling algorithm \cite{AIY13}.
Recall that our approach differs from the pruned labeling algorithm by only performing a thorough BFS for a small set of vertices, while running a lightweight local ball growing procedure for most vertices.
We found that this simple modification leads to 1.5-2.5x reduction in number of landmarks stored.
The preprocessing time is reduced by 2-15x as well.
While our algorithm does not always output the exact distance like the pruned labeling algorithm, we observe that the stretch is at most $1\%$, relative to the average distance.

Next we compare our approach to two approximate distance sketches with strong theoretical guarantees, Das Sarma et al. sketch \cite{BG12,DGNP10} and a variant of Thorup-Zwick's 3-approximate distance oracle \cite{TZ05}, which uses high degree vertices as global landmarks \cite{CSTW09}.
We observe that our approach incurs lower stretch and requires less space compared to Das Sarma et al. sketch.
The accuracy of Thorup-Zwick sketch is comparable to ours, but we require much fewer landmarks.

\subsection{Experimental Setup}

To ensure the robustness of our results, we measure performances on a diverse collection of directed and undirected graphs, with the datasets coming from different domains, as described by Table \ref{table:stat}.
Stanford, Google and BerkStan are all Web graphs in which edges are directed.
DBLP (collaboration network) and Youtube (friendship network) are both undirected graphs where there is one connected component for the whole graph.
Twitter is a directed social network graph with about $84\%$ vertices inside the largest strongly connected component.
All the datasets are downloaded from the Stanford Large Network Dataset Collection \cite{snapnets}.

\begin{table}[!hbt]
	\centering
	{\small\begin{tabular}{| l | l | l | c | c | c |}
		\hline
		graph & \# nodes & \# edges & category & type \\
		\hline
		DBLP & 317K & 1.0M & Collaboration & Undirected \\
		\hline
		Twitter & 81K & 1.8M & Social & Directed \\
		\hline
		Stanford & 282K & 2.3M & Web & Directed \\
		\hline
		Youtube & 1.1M & 3.0M & Social & Undirected \\
		\hline
		Google & 876K  & 5.1M & Web & Directed \\
		\hline
		BerkStan & 685K & 7.6M & Web & Directed \\
		\hline
	\end{tabular}}
	\caption{Datasets used in experiments.}
	\label{table:stat}
\end{table}

\vspace{-0.15in}
\noindent{\bf Implementation details:}
We implemented all four algorithms in Scala, based on a publicly available graph library.\footnote{\url{https://github.com/teapot-co/tempest}}
The experiments are conducted on a 2.30GHz 64-core Intel(R) Xeon(R) CPU E5-2698 v3 processor, 40MB cache, 756 GB of RAM.
Each experiment is run on a single core and loads the graph into memory before beginning any timings.
The RAM used by the experiment is largely dominated by the storage needed for the landmark sets.

\begin{figure*}[!hbt]
	\centering
	\begin{subfigure}{.333\textwidth}
		\centering
		\includegraphics[width=1.1\linewidth]{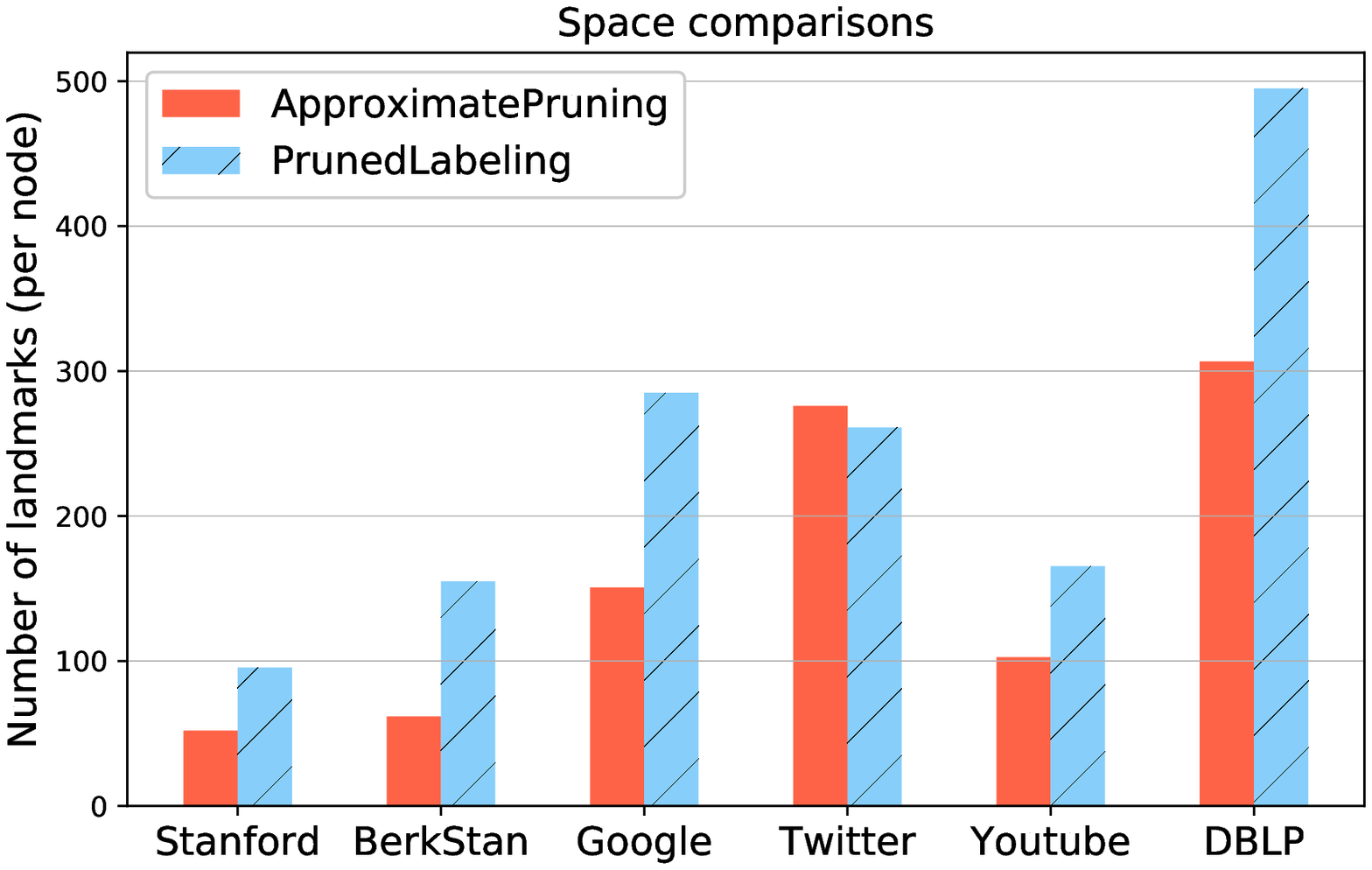}
	\end{subfigure}%
	\begin{subfigure}{.333\textwidth}
		\centering
		\includegraphics[width=1.1\linewidth]{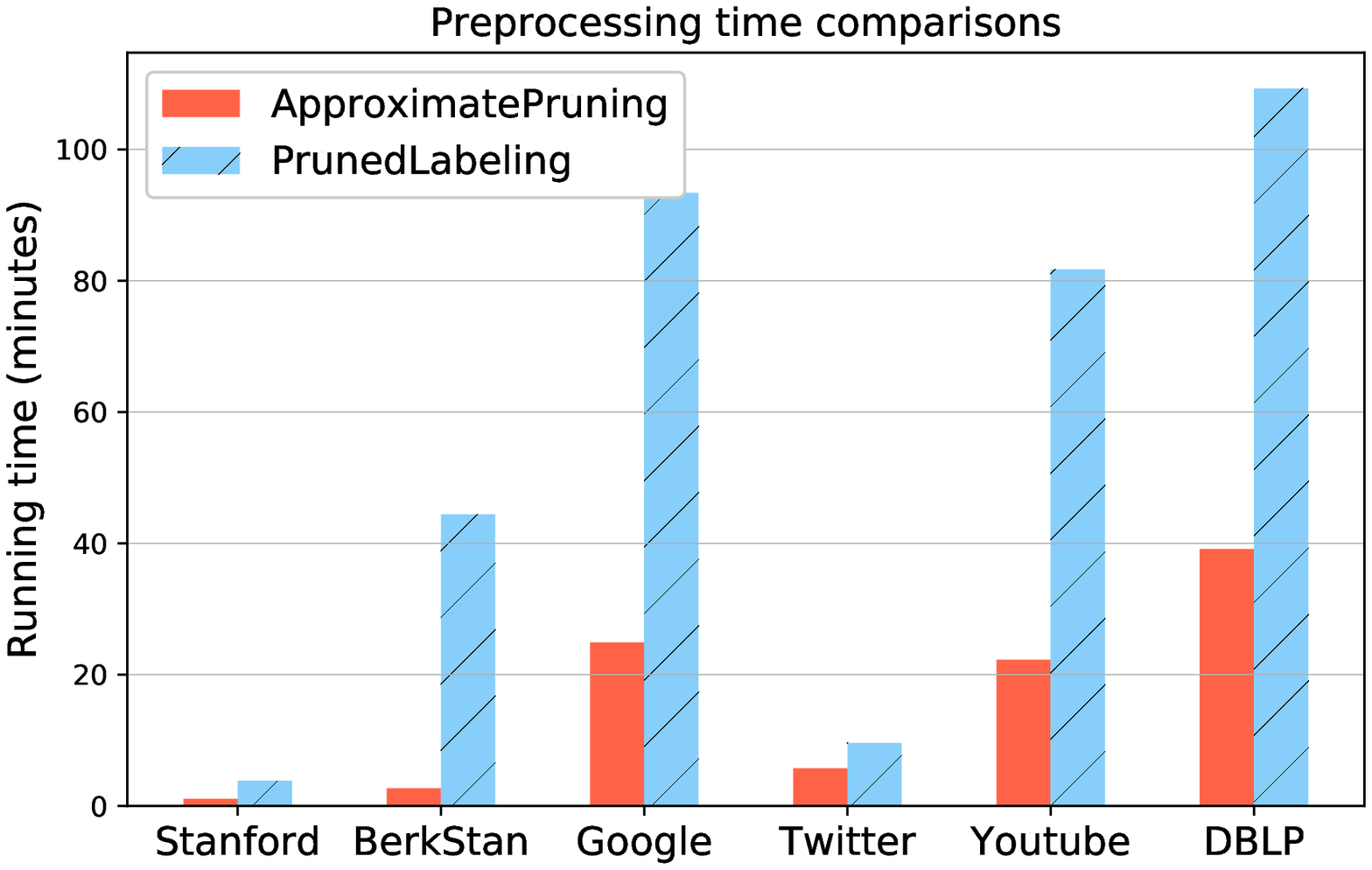}
	\end{subfigure}%
	\begin{subfigure}{.333\textwidth}
		\centering
		\includegraphics[width=1.1\linewidth]{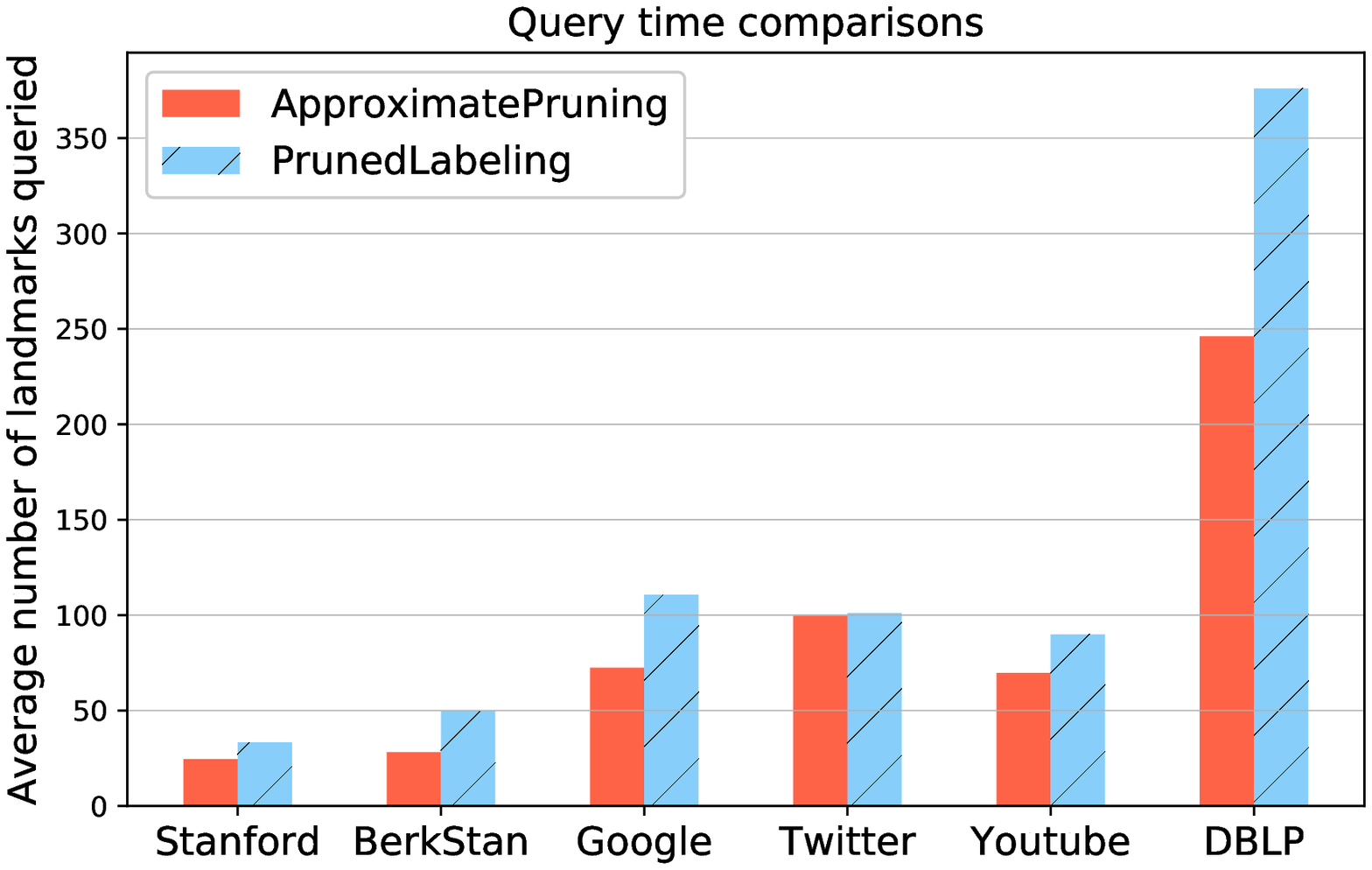}
	\end{subfigure}
	\vspace{-0.2in}
	\caption{Comparing the efficiency of our approach to the pruned labeling algorithm.}
	\label{fig_exact}
\end{figure*}

\begin{table*}[!hbt]
	\centering
	\begin{tabular}{| l | l | l | l | l | l | l |}
		\hline
		& Stanford & BerkStan & Google & Twitter & Youtube & DBLP \\
		\hline
		Relative Average Stretch & 0.37\% &
		0.20\%  &
		0.51\%  &
		0.29\%  &
		0.33\%  &
		1.1\%   \\
		\hline
		Maximum Relative Stretch & 21/10 &
		10/7 &
		8/5 &
		4/3 &
		4/3 &
		7/5 \\
		\hline
		Average Additive Stretch &  0.046 &
		0.030  &
		0.060  &
		0.014  &
		0.018  &
		0.075  \\
		\hline
		Maximum Additive Stretch & 11 &
		3 &
		3 &
		1 &
		2 &
		2 \\
		\hline
		Average Distance & 12.3 &
		14.6 &
		11.7 &
		4.9 &
		5.3 &
		6.8 \\
		\hline
	\end{tabular}
	\caption{Measuring the accuracy of our approach.}
	\label{table_exact}
	\vspace{-0.10in}
\end{table*}

\medskip
\noindent{\bf Parameters:}
In the comparison between the pruned labeling algorithm and our approach, we order the vertices in decreasing order by the indegree plus outdegree of each vertex.%
\footnote{There are more sophisticated techniques such as ordering vertices using their betweenness centrality scores \cite{DGPW14}.
It is conceivable that our algorithm can be combined with such techniques.}
Recall that there are two input parameters used in our approach, the number of global landmarks $H$ and the radiuses of local balls $\set{l_i}_{i = H+1}^n$.
To tune $H$, we start with 100, then keep doubling $H$ to be 200, 400, etc.
The radiuses $\set{l_i}_{i \ge H}$ are set to be $2$ for all graphs.%
\footnote{It follows from our theoretical analysis that the radiuses should be less than half of the average distance.
As a rule of thumb, setting the radius as 2 works based on our experiments.}

\medskip
\noindent{\bf Benchmarks:}
For the Thorup-Zwick sketch, in the first step, $H = \sqrt {n}$ vertices are sampled uniformly at random as global landmarks.
In the second step, every other vertex grows a local ball as its landmark set until it hits any of the $\sqrt {n}$ vertices.
All vertices within the ball are used as landmarks.
This method uses $O(n^{3/2})$ landmarks and achieves $3$-stretch in worst case.
In the follow up work of Chen et al. \cite{CSTW09}, the authors show a variant which uses high degree vertices as global landmarks and observe better performance.
We implement Chen et al.'s variant in our experiment, and use the $H$ vertices with the highest indegree plus outdegree as global landmarks.
In the experiment, we start with $H$ equal to $\sqrt{n}$.
Then we report results for $\sqrt n$ multiplied by $\set{2, 1/2, 1/4, 1/8}$.

For the Das Sarma et al. sketch, first, $\log n$ sets $S_i$ of different sizes are sampled uniformly from the set of vertices $V$, for $0 \le i < \log n$, where the size of $S_i$ is $2^i$.
Then a breadth first search is performed from $S_i$, so that every vertex $x \notin S_i$ finds its closest vertex inside $S_i$ in graph distance.
This closest vertex is then used as a landmark for $x$.
The number of landmarks used in Dar Sarma's sketch is $n \log n$, and the worst case multiplicative stretch is $\log n$.
If more accurate estimation is desired, one can repeat the same procedure multiple times and union the landmark sets together.
We begin with 5 repetitions, then keep doubling it to be 10, 20 etc.

Our approach differs from the above two methods by using the idea of pruning while running BFS.
This dramatically enhances performance in practice, as we shall see in our experiments.

\medskip
\noindent{\bf Metrics:} We measure the stretch of the estimated distances, and compute aggregated statistics over a large number of queries.
For a query $(x, y)$, if $y$ is reachable from $x$, but the algorithm reports no common landmark between the landmark sets of $x$ and $y$, then we count such a mistake as a ``False disconnect error.''
On the other hand, if $y$ is not reachable from $x$, then it is not hard to see that our algorithm always reports correctly that $y$ is not reachable from $x$.
In the experiments, we compute $\dist(x, y)$ using Dijkstra's algorithm.

To measure space usage, we report the number of landmarks per node used in each algorithm as a proxy.
Since the landmark sets are stored in Int to Float hash maps, the actual space usage would be eight bytes times the landmark sizes in runtime, with a constant factor overhead.

For the query time, recall that for each pair of vertices $(x, y)$, we estimate their distance by looking at the intersection of $L_F(x)$ and $L_B(y)$ and compute the minimum interconnecting distance (c.f. equation \ref{eq_query}).
To find the minimum, we iterate through the smaller landmark set.
Hence the running time is $\min(\abs{L_F(x)}, \abs{L_B(y)})$ multiplied by the time for a hash map lookup, which is a small fixed value in runtime.
A special case is when $y \in L_F(x)$ or $x \in L_B(y)$, where only one hash map lookup is needed.
We will report the number of hash map lookups as a proxy for the query time.%
\footnote{It is conceivable that more sophisticated techniques may be devised to speedup set intersection.
We leave the question for future work.}

\begin{figure*}
	\centering
	\begin{subfigure}{.333\textwidth}
		\centering
		\includegraphics[width=1.1\linewidth]{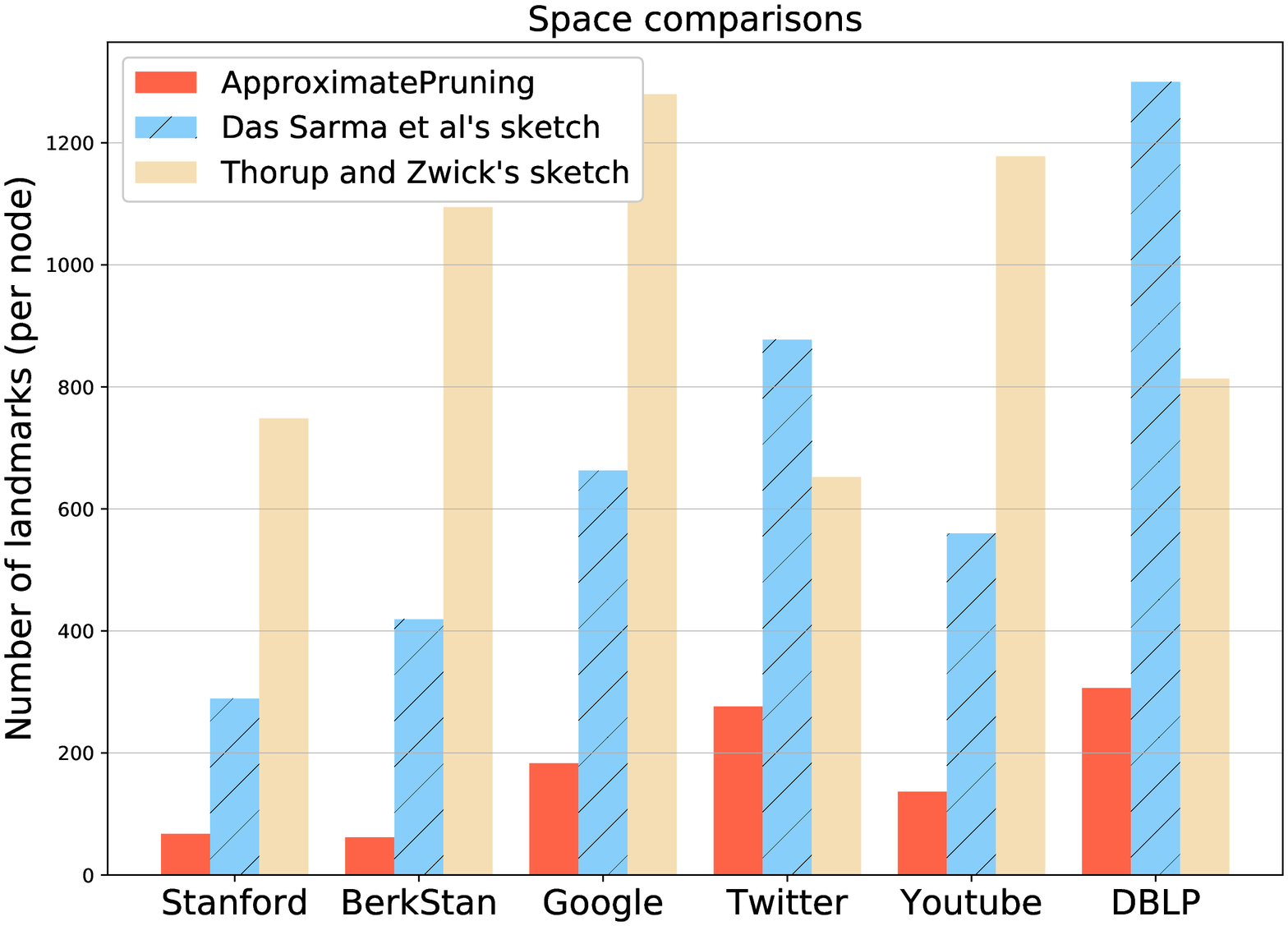}
	\end{subfigure}%
	\begin{subfigure}{.333\textwidth}
		\centering
		\includegraphics[width=1.1\linewidth]{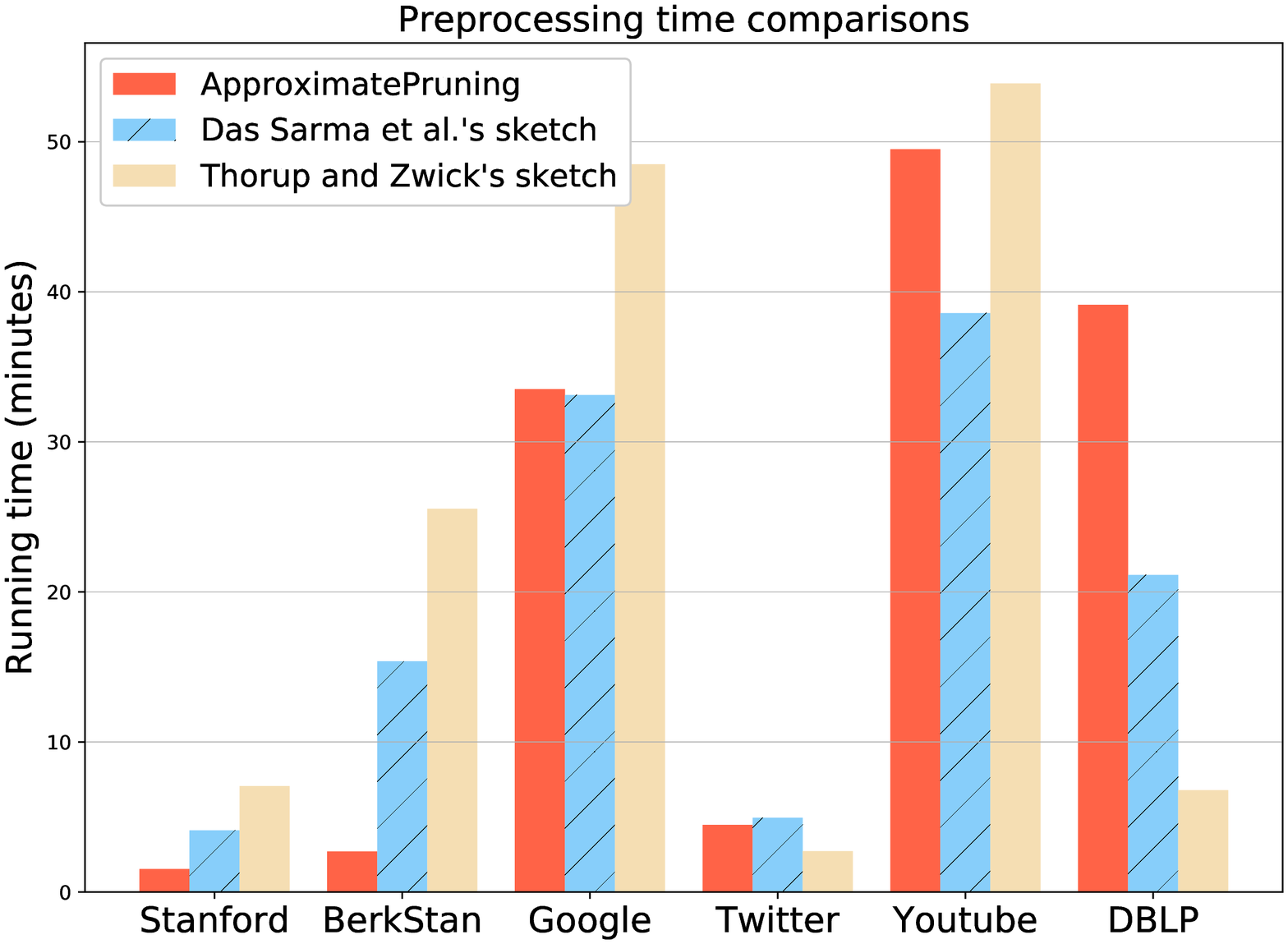}
	\end{subfigure}%
	\begin{subfigure}{.333\textwidth}
		\centering
	\includegraphics[width=1.1\linewidth]{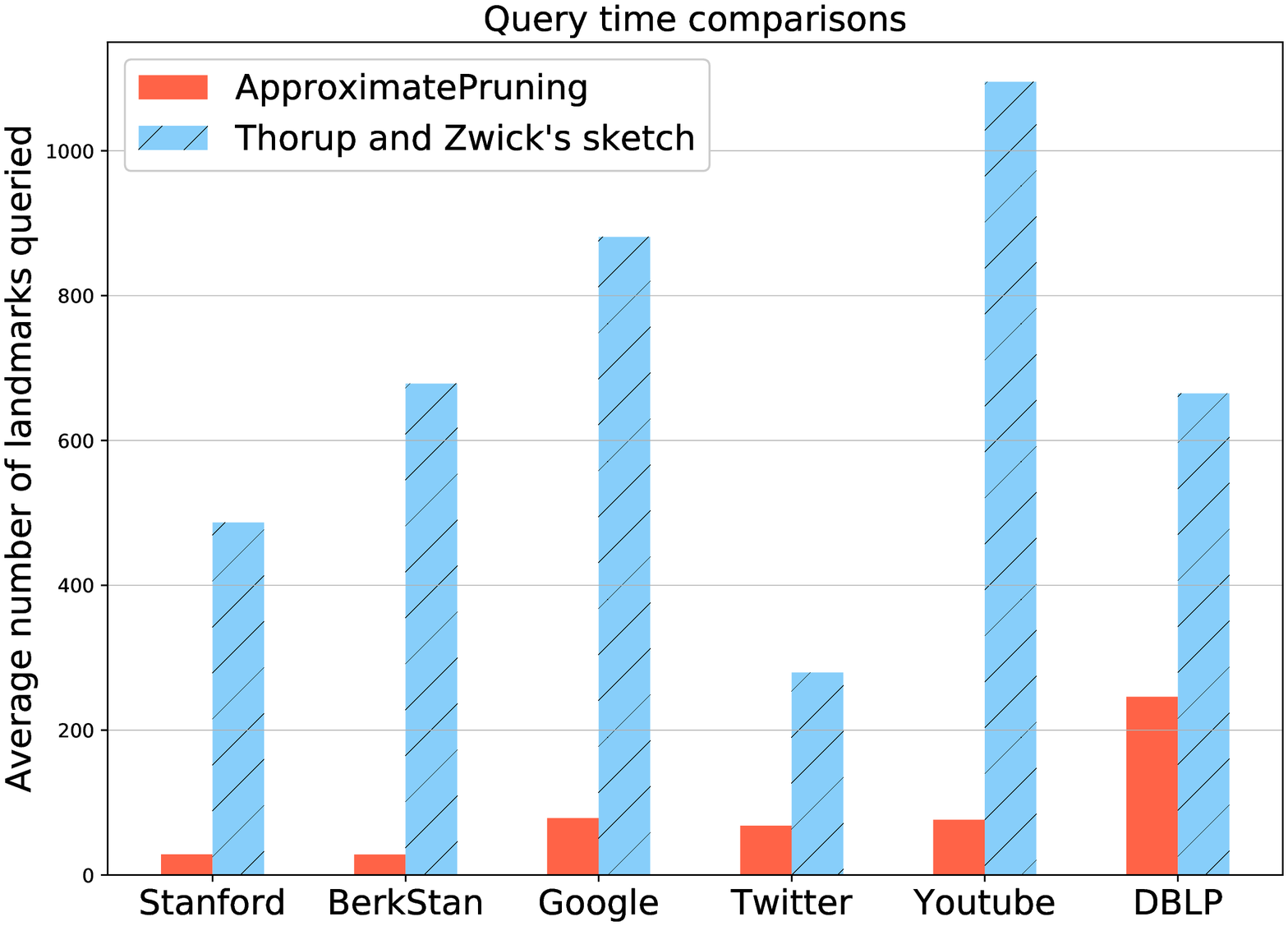}
	\end{subfigure}
	\vspace{-0.2in}
	\caption{Comparing the efficiency of our approach to two well known distance sketches with strong theoretical guarantees.}
	\label{fig_approx}
\end{figure*}

\begin{table*}[!hbt]
	\begin{minipage}{0.64\textwidth}
	\centering
	\begin{tabular}{| l | l | l | l | l | l | l |}
	\hline
			& Stanford & BerkStan & Google & Twitter & Youtube & DBLP \\
	\hline
	DS et al. sketch & 20.9\% &
	17.2\% &
	21.1\% &
	11.1\% &
	13.1\% &
	5.4\% \\
	\hline
	TZ sketch & 0.30\% &
	0.21\% &
	0.36\% &
	1.65\% &
	0.03\% &
	2.15\% \\
	\hline
	Our approach & 0.16\% &
	 0.20\%  &
	 0.22\%  &
	 0.29\%  &
	 0.04\%  &
	 1.1\%   \\
	\hline
	\end{tabular}
	\caption{Measuring the stretch for all three methods.}
	\label{table_approx}
	\end{minipage}%
	\begin{minipage}{0.36\textwidth}
		\begin{tabular}{| l | l | l | l | l | l |}
			\hline
			Youtube	&  ${\sqrt n} / 2$ & ${\sqrt n} / 4$ & ${\sqrt n} / 8$ & Ours \\
			\hline
			Stretch & 0.04\% & 0.11\% & 0.07\% & 0.04\% \\
			\hline
			\# Landmarks & 731 & 648 & 811 & 137 \\
			\hline
			Preprocessing & 37m & 31m & 35m & 50m \\
			\hline
		\end{tabular}
		\caption{Varying $H$ in TZ sketch.}
		\label{table_youtube}
	\end{minipage}
	\vspace{-0.10in}
\end{table*}

\subsection{Comparisons to Exact Methods}

We report the results comparing our approach to the pruned labeling algorithm.
The pruned labeling algorithm is exact.
To measure the accuracy of our approach, we randomly sample  $2000$ pairs of source and destination vertices.
The number of global landmarks is set to be 400 for the Stanford dataset, 1600 for the DBLP dataset, and 800 for the rest of the datasets.

Figure \ref{fig_exact} shows the preprocessing time, the number of landmarks and average query time used by both algorithms.
We see that our approach reduces the number of landmarks used by 1.5-2.5x, except on the Twitter dataset.%
\footnote{By setting the radiuses $\set{l_i}$ to be 1, we  incur $0.72\%$ relative additive stretch by using 173 landmarks per node, which improves over the pruned labeling algorithm by 1.5x.}
Our approach performs favorably in terms of preprocessing time and query time as well.

The accuracy of our computed estimate is shown in Table \ref{table_exact}.
We have also measured the median additive stretch, which turns out to be zero in all the experiments.
To get a more concrete sense of the accuracy measures, consider the Google dataset as an example.
Since the average additive stretch is $0.06$ and there are 2000 pairs of vertices, the total additive stretch is at most 120 summing over all 2000 pairs!
Specifically, there can be at most 120 queries with non-zero additive stretch and for all the other queries, our approach returns the exact answer.
Meanwhile, among all the datasets, we observed only one ``False disconnect error'' in total.
It appeared in the Stanford Web graph experiment, where the true distance is 80.

\subsection{Comparisons to Approximate Methods}

Next we compare our approach to Das Sarma et al.'s sketch (or DS et al. sketch in short) and the variant of Thorup and Zwick's sketch (or TZ sketch in short).
Similar to the previous experiment, we sample 2000 source and destination vertices uniformly at random to measure the accuracy.

We start by setting the number of global landmarks to $\sqrt{n}$ in Thorup-Zwick sketch.
To allow for a fair comparison, we tune our approach so that the relative average stretch is comparable or lower.
Specifically, the Stanford, BerkStan and Twitter datasets use $H = 800$, the Google and DBLP datasets use $H = 1600$ and the Youtube dataset uses $H = 3200$.

Figure \ref{fig_approx} shows the number of landmarks needed in each algorithm as well as the amount of preprocessing time consumed.
Overall, our approach uses much fewer landmarks than the other two algorithms.
In terms of preprocessing time, our approach is comparable or faster on all datasets, except on the DBLP network.
We suspect that this may be because the degree distribution of the DBLP network is flatter than the others.
Hence performing the pruning procedures on a small subset of high degree vertices are less effective in such a scenario.

\begin{figure*}
	\centering
	\begin{minipage}{.666\textwidth}
		\vspace{-0.045in}
		\centering
		\begin{subfigure}{0.5\linewidth}
			\includegraphics[width=0.95\linewidth]{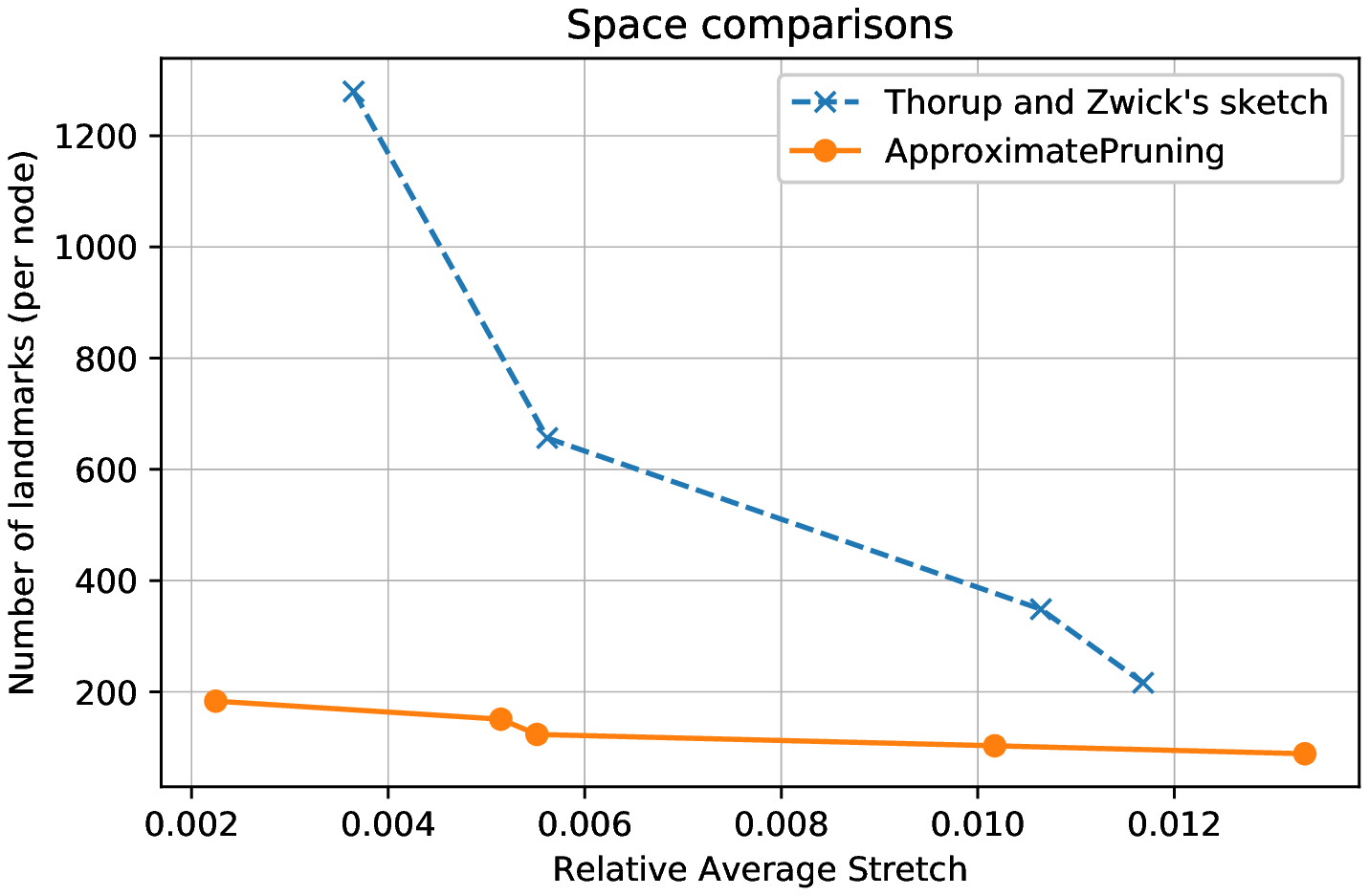}
		\end{subfigure}%
		\begin{subfigure}{0.5\linewidth}
			\includegraphics[width=0.95\linewidth]{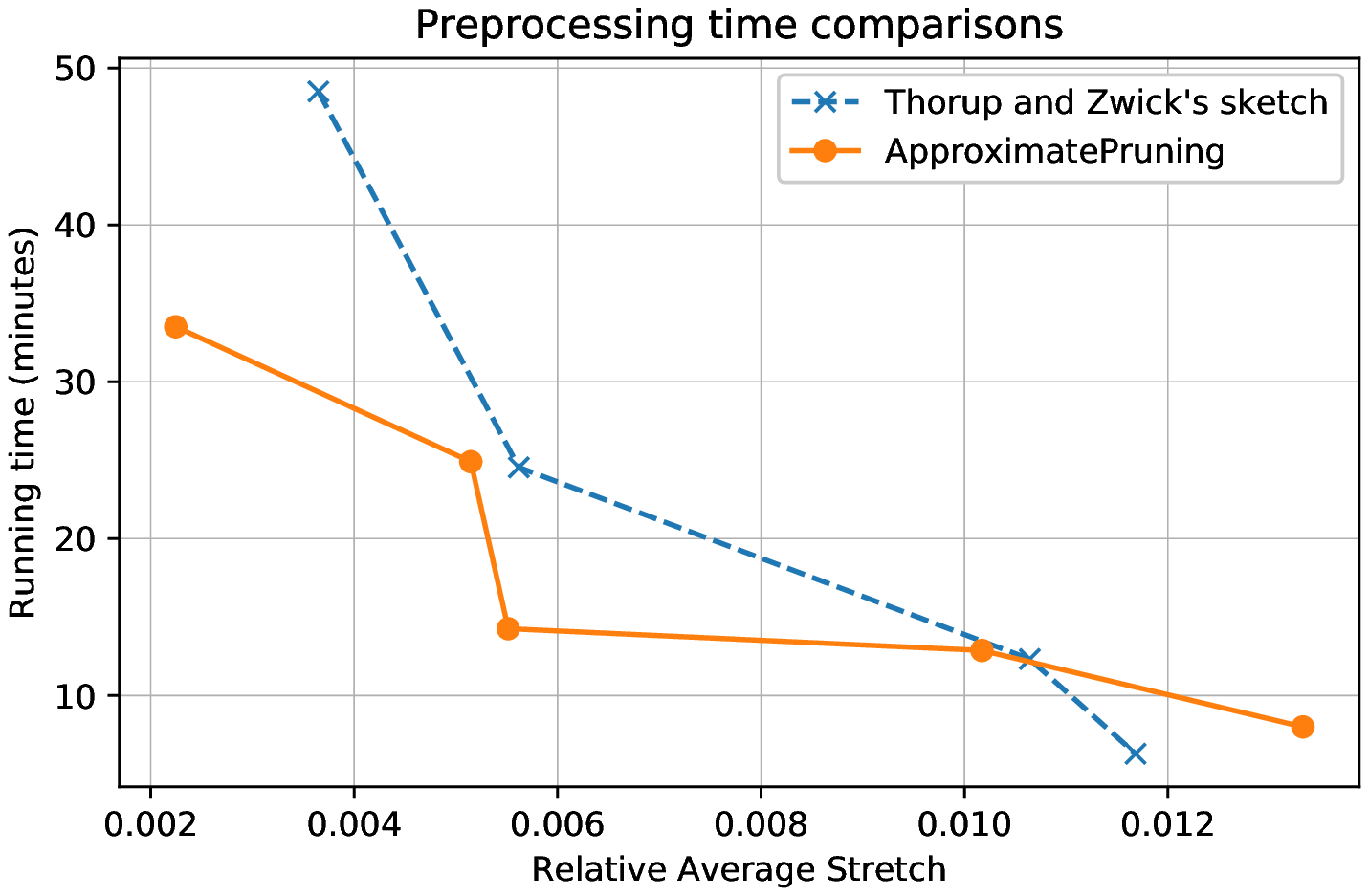}
		\end{subfigure}
		\vspace{-0.165in}
		\caption{Varying $H$ in TZ sketch and our approach, on the Google dataset.}
		\label{fig_tune}
	\end{minipage}%
	\begin{minipage}{0.333\textwidth}
		\centering
		\includegraphics[width=0.95\linewidth]{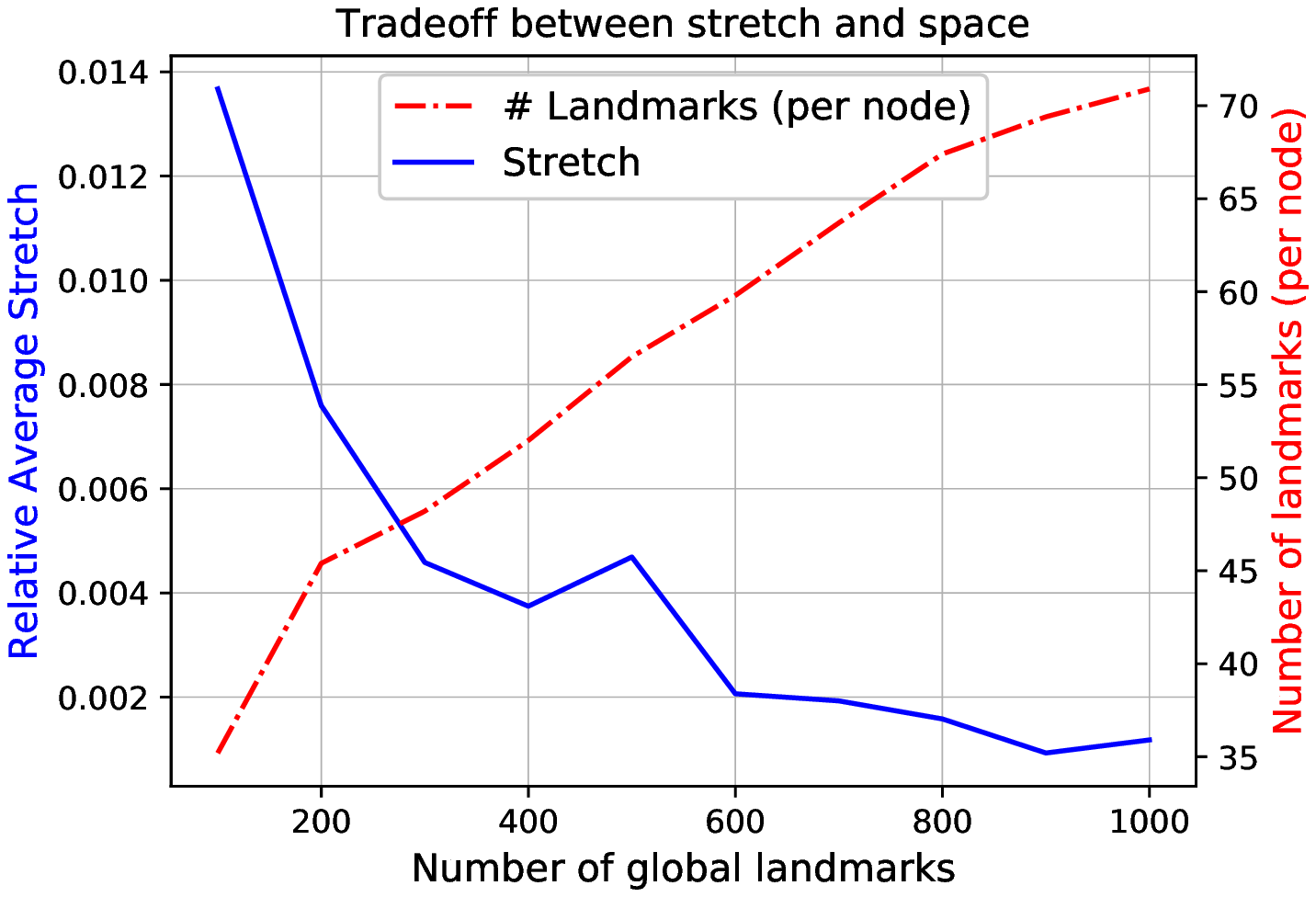}
		\vspace{-0.165in}
		\caption{Tradeoff on the Stanford dataset.}
		\label{fig_tradeoff}
	\end{minipage}
\end{figure*}

We next report the relative average stretch for all three methods.
As can be seen in Table \ref{table_approx}, our approach is comparable to or slightly better than Thorup and Zwick's sketch, but much more accurate than Das Sarma et al's sketch.
Note that the latter performed significantly worse than the other two approaches.
We suspect that this may be because the sketch does not utilize the high degree vertices efficiently.
Lastly, our approach performs favorably in the query time comparison as well.
The query time of Das Sarma et al.'s sketch are not reported because of the worse accuracy.

\medskip
\noindent{\bf Effect of parameter choice:} Note that in the above experiment, for Thorup and Zwick's sketch, we have set the number of global landmarks $H$ to be $\sqrt n$.
In the next experiment, we vary the value of $H$ to $\sqrt n$ multiplied by $\set{2, 1/2, 1/4, 1/8}$.

First, we report a detailed comparison on the Google dataset in Figure \ref{fig_tune}.
Note that when $H = 2\sqrt n$, the Thorup and Zwick's sketch requires over 2000 landmarks per node which is significantly larger than the other values.
Hence, we dropped the data point from the plot.
For our approach, we double $H$ from 100 up to 1600.
Overall, we can see that our approach requires fewer landmarks across different stretch levels.

Next, we report brief results on the Youtube dataset in Table \ref{table_youtube} since the results are similar.
The conclusions obtained from other datasets are qualitatively similar, and hence omitted.

\subsection{More Experimental Observations}

By varying the number of global landmarks used Algorithm \ref{alg_prune}, it is possible to obtain a smooth tradeoff between stretch and number of landmarks used.
As an example, we present the tradeoff curve for the Stanford Web dataset in Figure \ref{fig_tradeoff}.
Here we vary the number of global landmarks used from 100 to 1000.
As one would expect, the relative average stretch decreases while the number of landmarks stored increases.

\section{Fundamental Limits of Landmark Sketches}\label{sec_lb}

This section complements our algorithm with lower bounds.
We begin by showing a matching lower bound for \ER~graphs, saying that any 2-hop cover needs to store at least $\tilde{\Omega}(n^{3/2})$ landmarks.
The results imply that the parameter dependence on $n$ of our algorithm is tight for \ER~graphs and random power law graphs with power law exponent $\beta > 3$.
It is worth mentioning that the results not only apply to landmark sketches, but also work for the family of labeling schemes.
Recall that labeling schemes associate a labeling vector for each vertex.
To answer a query for a pair of vertices $(x, y)$, only the labeling vectors of $x, y$ are accessed.
We first state the lower bound for \ER~graphs.

\begin{theorem}\label{thm_sp3}
Let $G = (V, E)$ be an \ER~graph where every edge is sampled with probability $p = 2\log n / n$.
With high probability over the randomness of $G$, any labelings which can recover all pairs distances exactly have total length at least ${\Omega}(n^{3/2} / \log^4 n)$.

In particular, any 2-hop cover needs to store at least ${\Omega}(n^{3/2} / \log^4 n)$ many landmarks with high probability.
\end{theorem}

For a quick overview, we divide $V$ into $\sqrt {n}$ sets of size $\sqrt n$ each.
We will show that the total labeling length for each set of $\sqrt n$ vertices has to be at least $\tilde{\Omega}(n)$.
By union bound over all the $\sqrt{n}$ sets, we obtain the desired conclusion.
We now go into the proof details.

\begin{proof}
	Denote by $r = np$.
	Let $d = \floor{\frac {\log n} {2 \log (np)}} - c$, where $c$ is a fixed constant (e.g. $c = 2$ suffices).
	Divide $V$ into groups of size $\sqrt n$.
	Clearly, there are $\sqrt n$ disjoint groups --
	let $S$ be one of them.
	Denote by $c_1$ a fixed constant which will be defined later.
	We argue that
	\begin{align}\label{eq_er_group}
		\Pr[\text{The total label length of $S$} \le c_1 \cdot r^{1 - 2c} n] \lesssim r^{1- 2c}.
	\end{align}
	Hence by Markov's inequality, with high probability except for $(\log n) r^{1- 2c} \sqrt n$ groups, all the other groups will have label length at least $c_1 \cdot r^{1-2c} n \gtrsim \tilde{\Omega}(n)$, because $r \le 2\log n$.
	Hence we obtain the desired conclusion.
	For the rest of the proof, we focus on proving equation \eqref{eq_er_group} for the group $S$.

	Let $\set{x_1, x_2, \dots, x_{\abs{S}}}$ be an arbitrary ordering of $S$.
	We grow the neighborhood of each vertex in $S$ one by one, up to level $d$.
	Denote by $G_1 = (V_1, E_1)$, where $V_1 = V$ and $E_1 = E$.
	For any $i \ge 1$, if $x_i \in V_i$, then we define
	define $T(x_i)$ to be the set of of vertices in $V_i$ whose distance is at most $d$ from $x_i$.
	Define $L(x_i) \subseteq T(x_i)$ to be the set of vertices in $G_i$ whose distance is equal to $d$ from $x_i$.
	On the other hand if $x_i \notin V_i$, then $T(x_i)$ and $L(x_i)$ are both empty.
	More formally,
	\begin{align*}
		&T(x_i) := \begin{cases}
		           \set{y : \dist_{G_i}(x_i, y) \le d}, &\mbox{if } x_i \in V_i \\
		           \varnothing, &\mbox{otherwise.}
		\end{cases} \\
			&L(x_i) := \set{y \in T(x_i) : \dist_{G_i}(x_i, y) = d}
	\end{align*}

	We then define $F_i = \cup_{j=1}^i T(x_j)$.
	Denote by $G_{i+1}$ to be the induced subgraph of $G_i$ on the remaining vertices $V_{i+1} = V \backslash F_i$.
	We show that with high probability, a constant fraction of vertices $x_i \in S$ satisfy that $\abs{L(x_i)} \geq \Omega((np)^d)$.

	\begin{lemma}[Martingale inequality]\label{lem_martingale_er}
		In the setting of Theorem \ref{thm_sp3}, with high probability, at least $\abs{S} / 2$ vertices $x_i \in S$ satisfy that $\abs{L(x_i)} \ge r^{d} / 6$.
	\end{lemma}
	\begin{proof}
		For any $1 \le i \le \abs{S}$, consider
	\[ X_i := \begin{cases}
		          1 & \textrm{if } x_i \notin V_i, \textrm{ or } \abs{F_{i-1}} > \abs{S} \cdot r^d \log n,
		          \textrm{or } \abs{L(x_i)} \geq r^{d} / 6 \\
		          0 & \textrm{otherwise.}
	\end{cases} \]
		We claim that $\Pr[X_i=1\mid X_1,\ldots,X_{i-1}]$ with high probability.
		It suffices to consider the case $x_i \in V_i$ and $\abs{F_{i-1}} \le \abs{S} r^d \log n$.
		It is not hard to verify that $\abs{F_{i-1}} \le n / \log n$ by our setting of $d$.
		Hence the size of $V_i$ is at least $n(1 - 1 / \log n)$.
		Note that the subgraph $G_i$ is still an \ER~random graph, and the number of vertices is at least $n(1 - 1 / \log n)$.
		By Fact \ref{fact_er}c), the size of $L(x_i)$ is at least
		\[ \frac 1 2 r^d (1 - \log^{-1} n)^d \ge r^d / 6, \]
		since $d \le \log n$.

		Thus by Azuma-Hoeffding inequality,
	$\sum_{i=1}^{\abs S} X_i \geq 0.99 \abs{S}$ with high probability.
		We will show below that the contributions to $\sum_{i=1}^{\abs S} X_i$ from $x_i \notin V_i$ and $\abs{F_{i-1}} > \abs{S} r^d \log n$ is less than $0.02 \abs{S}$.
		Hence by taking union bound, we obtain the desired conclusion.

		First, we show that the number of $x_i$ such that $x_i \notin V_i$ are at most $0.01 |S|$ with high probability.
		Note that $x_i \notin V_i$ implies that there exists some vertex $x_j$ with $j < i$ such that $\dist(x_i, x_j) \le d$.
		On the other hand, by Fact \ref{fact_er},
		\[ \Pr[\dist(x, y) \le d] \le \frac {3 r^d} n, \forall x, y \in S. \]
		Hence, it is not hard to verify that the expected number of vertex pairs in $S$ whose distance is at most $d$,	is $O(\abs{S}^2 r^{2d} / n) \lesssim \abs{S} / \log n$, by the setting of $d$.
	By Markov's inequality, with probability $1 - 1 / \log n$ only $0.01 \abs{S}$	vertex pairs have distance at most $d$ in $S$.
	Hence there exists at most $0.01 \abs{S}$ $i$'s such that $x_i \notin V_i$.

		Secondly for all $1 \le i \le \abs{S}$, the set of vertices $T_i$ is a subset of $N_{d}(x_i)$, the set of vertices within distance $d$ to $x_i$ on $G$.
		By Fact \ref{fact_er}c), the size of $N_d(x_i)$ is at most $2 r^d$.
		Hence we have $\abs{F_i} \le 2 \abs{S} r^d$ for all $1 \le i \le \abs{S}$ with high probability.
		This proves the Lemma.
	\end{proof}
	Now we are ready to finish the proof.
	Given the labels of $S$, we can recover all pairwise distances in $S$.
	Let $\dist_S: S\times S \rightarrow \mathbb{N}$ denote the distance function restricted to $S$.
	Consider the following:
	\begin{enumerate}
		\item[a)]	$\exists \abs{S}^2 / 9$ pairs $(x_i,x_j)$ such that $\dist_S(x_i,x_j)\leq 2d+1$.
		We know by Fact \ref{fact_er} that $\Pr[\dist(x_i, x_j) \le 2d+1] \le 2 r^{2d+1} / n$, for any $x_i, x_j \in S$.
		Hence the expected number of pairs with distance at most $2d+1$ in $S$,
		is at most $2 \abs{S}^2 \cdot r^{2d+1}/n \lesssim r^{1-2c} n$.
		By Markov's inequality, the probability that a random graph induces any such distance function is $r^{1-2c} n / (\abs{S}^2 / 8) \lesssim r^{1 - 2c}$.

		\item[b)] The number of pairs such that	$\dist_S(x_i,x_j)\leq 2d+1$ is at most $\abs{S}^2 / 8$ in $S$.
		Let $A$ denote \[ \left\{(x, y) \in S \times S \mid \dist(x, y) > 2d+1, \mbox{ and } \abs{L(x)}, \abs{L(y)} \ge r^{d}/ 6\right\}. \]
		By Lemma \ref{lem_martingale_er} and our assumption for case b), the size of $A$ is at least
		$\binom{\abs{S} / 2}{2} - \abs{S}^2 / 8 \ge \abs{S}^2 / 9.$
		For any $(x, y) \in A$, $L(x)$ and $L(y)$ are clearly disjoint.
		Note that the event whether there exists an edge between $L(x)$ and $L(y)$ is independent, conditional on revealing the subgraph for all $x \in S$ up to distance $d$. Hence
		\begin{align*}
			&\Pr\left[\dist_S(x, y) > 2d+1,~\forall (x, y) \in A\right]\\
			&\leq \prod_{(x, y) \in A} \Pr\left[\mbox{there is no edge between } L(x) \mbox{ and } L(y) \right] \\
			&\leq \prod_{(x, y) \in A} (1 - p)^{\abs{L(x)} \times \abs{L(y)}} \\
			&\le \exp\left(-p \times \abs{A} \times r^{2d} / 72 \right) \tag{because $\abs{L(x)}, \abs{L(y)} \ge r^d / 6$} \\
			&\leq \exp(-c_1 r^{1 - 2c} n). \tag{because $\abs{A} \ge  n / 9$}
		\end{align*}
		where $c_1 = 72 \times 9$ in the last line.
		Denote by $\kappa = c_1 \cdot r^{1-2c} n$.
		Note that the number of labelings of length (or number of bits) less than $\kappa$ is at most $2^{\kappa}$.
		For each labeling, the probability that it correctly gives all pairs distances is at most $\exp(-\kappa)$ by our argument above.
		Therefore by union bound, the probability that the total labeling length of $\abs S$ is at most $\kappa$ is at most $2^{\kappa} \cdot \exp(\kappa) \le r^{1- 2c}$ for large enough $n$.
	\end{enumerate}
	To recap, by combining case a) and b), we have shown that equation \eqref{eq_er_group} is true.
	Hence the proof is complete.
\end{proof}

\smallskip
\noindent{\bf Extensions to $\beta > 3$:}
It is worth mentioning that the lower bound on \ER~graphs can be extended to random power law graphs with $\beta > 3$.
The proof structure is similar because the degree distribution has finite variance, hence the number of high degree vertices is small.
The difference corresponds to technical modifications which deal with the neighborhood growth of random graphs with constant average degree.
We state the result below and leave the proof to Appendix \ref{sec_lb3}.
\begin{theorem}\label{thm_lb3}
Let $G = (V, E)$ a random power law graph with average degree $\avgdeg > 1$ and exponent $\beta > 3$.
	With high probability over the randomness of $G$, any labelings which can recover all pairs distances exactly have total length at least $\tilde{\Omega}(n^{3/2})$.

	In particular, any 2-hop cover needs to store at least $\tilde{\Omega}(n^{3 / 2})$ many landmarks with high probability.
\end{theorem}


\noindent{\bf Lower bounds for $\beta$ close to 2:}
Next we show that the parameter dependence of our algorithm is tight when $\beta$ is close to 2.
Specifically, any 2-hop cover needs to store at least $\Omega(n^{3/2 - \varepsilon})$ many landmarks when $\beta = 2 + \varepsilon$.
Hence it is not possible to improve over our algorithm when $\beta$ is close to 2.
Furthermore, the lower bound holds for the general family of labeling schemes as well.

\newcommand{\slow}{S_{\mathrm{low}}}
\newcommand{\shigh}{S_{\mathrm{high}}}
\begin{theorem}\label{thm_lb_2}
	Let $G = (V, E)$ a random power law graph with average degree $\avgdeg > 1$ and exponent $\beta = 2 + \varepsilon$ for $\varepsilon < 1/2$.
	With high probability over the randomness of $G$, any labelings which can recover all pairs distances exactly have total length at least $\Omega(n^{3/2 - \varepsilon})$.

	In particular, any 2-hop cover needs to store at least $\Omega(n^{3 / 2 - \varepsilon})$ many landmarks with high probability.
\end{theorem}

The proof is conceptually similar to Theorem \ref{thm_sp3}, so we sketch the outline and leave the proof to Appendix \ref{sec_pf_lb2}.

Let $\shigh$ be the set of vertices whose degrees are on the order of $\sqrt n$.
Let $\slow$ be a set of $\sqrt n$ vertices, where each vertex has weight between $\avgdeg$ and $2\avgdeg$.
Such a set is guaranteed to exist because there are $\Theta(n)$ of them.

We first reveal all edges of $G$ other than the ones between $\shigh$.
We show that at this stage, most vertices in $\slow$ are more than 3 hops away from each other.
If for some pair $(x, y)$ in $\slow$ whose distance is larger than three, and both $x$ and $y$ connect to exactly one (but different) vertex in $\shigh$,
then knowing whether $\dist(x, y) = 3$ will reveal whether their neighbors in $\shigh$ are connected by an edge.

Based on the observation, we show that the total labeling length of $\slow$ is at least $\tilde{\Omega}(n^{3 - \beta})$.
This is because the random bit between a vertex pair in $\shigh$ has entropy $\Omega(n^{2-\plexp})$.
Since there are $\Theta(n)$ pairs of vertices in $\shigh$, the entropy of the labelings of $\slow$ must be $\Omega(n^{3-\beta})$ (hence, its size must also be at least $\Omega(n^{3-\beta})$).
Similar to Theorem \ref{thm_sp3}, this argument is applied to $\sqrt{n}$ disjoint sets of ``$\slow$'', summing up to an overall lower bound of $\Omega(n^{7/2-\beta})=\Omega(n^{3/2-\epsilon})$.

\section{Conclusions and Future Work}\label{sec:discuss}

In this work, we presented a pruning based landmark labeling algorithm.
The algorithm is evaluated on a diverse collection of networks.  It demonstrates improved performances compared to the baseline approaches.
We also analyzed the algorithm on random power law graphs and \ER~ graphs.
We showed upper and lower bounds on the number of landmarks used for the \ER~ random graphs and random power law graphs.

There are several possible directions for future work.
One direction is to close the gap in our upper and lower bounds for random power law graphs.
We believe that any improved understanding can potentially lead to better algorithms for real world power law graphs as well.
Another direction is to evaluate our approach on transportation networks, which correspond to another important domain in practice.

\bigskip
\noindent{\bf Acknowledgements.} The authors would like to thank Fan Chung Graham, Tim Roughgarden, Amin Saberi and D. Sivakumar for useful feedback and suggestions at various stages of this work.
Also, thanks to the anonymous referees for their constructive reviews.
Hongyang Zhang is supported by NSF grant 1447697.

\bibliographystyle{ACM-Reference-Format}
\balance
\bibliography{rf}

\appendix

\onecolumn
\section{Proof of Theorem~\ref{thm_sp2_ub}: Upper Bounds for $2 < \beta < 3$}
\label{sec:pf2}

In this section, we present the proof of Theorem \ref{thm_sp2_ub}, which analyzes the performance of Algorithm \ref{alg_approx} on random power law graphs.
We show that Algorithm \ref{alg_approx} outputs a 2-hop cover in Proposition \ref{lem:acc2}.
Then we bound the landmark set sizes in Proposition \ref{lem:size2}.

We introduce a few notations first.
For a set of vertices $S \subseteq V$, let $d_S = \sum_{x \in S} d_x$ denote the sum of their degrees.
Denote by $x \sim y$ if there is an edge between $x, y$.
For two disjoint sets $S$ and $T$, denote by $S \sim T$ if there exists an edge between $S$ and $T$, and $S \nsim T$ if there does not exist any edge between $S$ and $T$.
We use $\starnode = \argmax_{x \in V}~p_x$ to denote the maximum weight vertex.
For any integer $1 \le i \le n-1$,
recall that $\level_i(x) = \set{y \in V : \dist(x, y) = i}$
denotes the set of vertices whose distance from $x$ is equal to $i$.
And $\neigh_i(x) = \set{y \in V : \dist(x, y) \le i}$ denotes the set of vertices whose distance from $x$ is at most $i$.
Let $\alpha_i(x)$ denote the number of edges between $N_i(x)$ and $V \backslash N_i(x)$.
Let $\vols{S}:=\sum_{x\in S} p_x^2$ denote the second moment of any $S \subseteq V$.

Throughout the section, we assume that $\seqdeg$ satisfies all the properties in Proposition \ref{prop:deg2} without loss of generality.

\begin{proposition}\label{lem:acc2}
	In the setting of Theorem \ref{thm_sp2_ub}, Algorithm \ref{alg_approx} returns a 2-hop cover $F(\cdot)$ with high probability.
\end{proposition}

\begin{proof}
	Recall that $l_i$ is the radius of the local ball from $x_i$.
	Denote by $l(x_i) = l_i$ for all $i \ge K+1$.
	Let $\Omega_S$ denote the set of graphs that satisfies
		\begin{align*}
			\level_{l(x)}(x) = \emptyset \text{ or } \dist(\starnode, x) \le l(x),
			\forall x \in V.
		\end{align*}
	We argue that Algorithm \ref{alg_approx} finds a 2-hop cover
	for any $G \in \Omega_S$, and
		\begin{align*}
			1 - \pr{\Omega_S} \le 2/n
		\end{align*}
	This would conclude the proof.

	We first argue that Algorithm \ref{alg_approx} is correct if $G \in \Omega_S$.
	Let $x$ and $y$ be two different vertices in $V$.
	If $x$ and $y$ are not reachable from each other, then clearly
	$F(x) \cap F(y) = \emptyset$.
	If $x$ and $y$ are reachable from each other, consider their distance
	$\dist(x, y)$.
	Note that when $\level_{l(x)}(x)$ (or $\level_{l(y)}(y)$) is empty,
	then $F(x)$ (or $F(y)$) includes the entire connected component that contains
	$x$ (or $y$). Therefore, $y \in F(x)$, vice versa.
	When none of them are empty,
	we know that $\dist(x, \starnode) \le l(x)$ and $\dist(y, \starnode) \le l(y)$ since $G \in \Omega_S$.
	We consider three cases:
	\begin{itemize}
		\item If $\dist(x, y) \le l(x) + l(y) - 2$, then there exists a node $z$ such
		that $\dist(x, z) \le l(x) - 1$ and $\dist(y, z) \le l(y) - 1$. By our
		construction, $z$ is in $F(x)$ and $F(y)$.
		\item If $\dist(x, y) = l(x) + l(y) - 1$, then consider the two nodes $z$ and $z'$ on one of the shortest path from $x$ to $y$,
		with $\dist(x, z) = l(x) - 1$ and $\dist(y, z) = l(y)$. If either $d_z$ or
		$d_{z'}$ is	at least $K$, then
		they have been added as a landmark to every node in $V$. Otherwise,
		assume without loss of generality that $d_z \ge d_{z'}$. Then our
		construction adds $z$ into $F(y)$ and clearly $z$ is also in
		$F(x)$, hence $z$ is a common landmark for $x$ and $y$.
		\item If $\dist(x, y) = l(x) + l(y)$, then clearly $\starnode$ is a common landmark for $x$ and $y$.
	\end{itemize}
	We now bound $1 - \pr{\Omega_S}$. Clearly,
	\begin{align*}
		&\ 1 - \pr{\Omega_S} \\
		\le&\ \sum_{x \in V} \pr{\level_{l(x)}(x) \neq
			\emptyset, \dist(\starnode, x) > l(x)} \\
		=&\ \sum_{x \in V}\sum_{k = 0}^{n-1} \pr{l(x) = k+1, \level_{k+1}(x) \neq \emptyset, \dist(\starnode, x) > k+1}
	\end{align*}
	Note that $l(x) = k+1$ and $\level_{k+1}(x) \neq \emptyset$ is the same as the event that:
	\begin{itemize}
		\item $\alpha_i(x) \le \delta n^{1 - \maxdeg}$, for $i = 0,\dots, k-1$;
		\item $\alpha_k(x) > \delta n^{1 - \maxdeg}$.
	\end{itemize}
	Hence,
	\begin{align}
		&\ \pr{l(x) = k+1, \level_{k+1}(x) \neq \emptyset,
				\dist(\starnode, x) > k+1} \nonumber \\
		\le &\ \pr{\alpha_k(x) > \delta n^{1 - \maxdeg},
				\dist(\starnode, x) > k+1} \nonumber \\
		\le &\ \pr{\alpha_k(x) > \delta n^{1 - \maxdeg},
				\vol{\level_k(x)} \le \frac{\delta n^{1 - \maxdeg}} 3}
					\label{eq:pr1} \\
		+		&\ \pr{\vol{\level_k(x)} > \frac{\delta n^{1 - \maxdeg}} 3,
				\dist(\starnode, u) > k+1} \label{eq:pr2}
	\end{align}
	For Equation \eqref{eq:pr1}, consider how $\alpha_k(x)$ is discovered when
	we do the level set expansion from node $x$.
	Conditioned on $a = \vol{\level_k(x)} \le \delta n^{1-\maxdeg} / 3$,
	$\alpha_k(x)$ is the sum of 0-1-2
	independent random variables, with expected value less than
	$\delta n^{1-\maxdeg} / 3$.
	Hence by Chernoff bound,
	Equation \eqref{eq:pr1} is at most
	$\exp(-\delta n^{1 - \maxdeg} / 6) \sim o(n^{-3})$.
	For Equation \eqref{eq:pr2}, conditioned on
	$\vol{\level_k(x)} \ge \delta n^{1 - \maxdeg} / 2$
	and $\starnode \notin \neigh_k(x)$,
	\begin{align*}
		\pr{\starnode \nsim \level_k(x)}
		\le \exp\left(-\frac{\delta n^{1 - \maxdeg} \weight {\starnode}} {2{\vol{V}}}\right)
		\sim o(n^{-3})
	\end{align*}
	The first inequality is because of Proposition \ref{prop:connect}.
	The second inequality is because ${\vol{V}} \sim \avgdeg n \pm o(n)$ by Proposition
	\ref{prop:deg2}.
	In summary, $1 - \pr{\Omega_S} \le 2 / n$.
\end{proof}

Next we consider the landmark set sizes.
There are three parts in each landmark set:
(1) the heavy nodes whose degrees are at least $K$;
(2) all the level sets before the last layer;
(3) the last layer that we carefully constructed.
It's not hard to bound the first part,
since the degree of a node is concentrated near its weight,
and it is not hard to show that the number of vertices whose weight is $\Omega(K)$ is ${\textsc O}(n K^{1-\plexp})$.
The second part can be bounded by the maximum number of layers,
 which is at most $O(\log n)$, the diameter of $G$.
For the third part, the idea is that before adding all the
nodes on the boundary layer, we already have a $(+1)$-stretch
scheme.
Therefore, for a given vertex $x$, it is enough if we only
add neighbors whose degree is bigger than themselves.
This reduces the amount
of vertices from $d_x$ to ${\textsc O}(d_x^{3 - \plexp})$, where $d_x$ denotes the degree of any vertex $x \in V$.

We first show that the volume of all the level sets
is at most $\textsc{O}(\delta n^{1 - \maxdeg})$
before the boundary layer in Proposition \ref{lem:volIter} and \ref{lem:pfp2}.
For the rest of the section, denote by $\alpha_k = \alpha_k(x)$ for any
$0 \le k \le n-1$, unless there is any ambiguity on the vertex we are considering.
Recall that $\alpha_k$ denotes the number of edges between
$\level_k(x)$ and $V \backslash \neigh_{k-1}(x)$.

\begin{proposition}\label{lem:volIter}
	Let $x$ be a fixed node.
	Let $k \lesssim \log n$.
	Let $\Omega_k$ denote the set of graphs	such that
		\[ \vol{\level_i(x)} < 4\delta n^{1 - \maxdeg},
			\text{ for any } 0 \le i \le k-1, \]
	and \[ \vol{\level_k(x)} > 4\delta n^{1 - \maxdeg}. \]
	Then $\pr{\alpha_k \le \delta n^{1 - \maxdeg} \mid \Omega_k}
		\le n^{-2}$.
\end{proposition}

\begin{proof}
	Let $a = \vol{\level_k(x)}$ and $b = \vol{\neigh_{k-1}(x)}$.
	Conditioned on $\Omega_k$,
		\[ a > 4\delta n^{1-\maxdeg} \text{ and } b \le 4k\delta n^{1-\maxdeg}. \]
	Clearly, the random variable $\alpha_k$ is the sum of
	independent 0-1 random variables. Let $\mu$ denote its expected value.
	For each $y \in \level_k(x)$, we know that $\w y \le a = O(\sqrt n)$.
	Let $\mu_y$ denote the expected number of edges between $y$ and
	$V \backslash \neigh_{k-1}(x)$, then
		\begin{align*}
			\mu_y &= \sum_{z : z \neq y \land z \notin \neigh_{k-1}(x)}
				\min(\frac {\w y \w z} {\vol{V}}, 1) \istrue{\w z \le \sqrt n}\\
				&\ge \w y (1 - \frac{b + \sum_{z \in V} \w z \istrue{\w z \ge \sqrt n}} {\vol{V}}) = \w y(1 - \kappa(n))
		\end{align*}
	because of Proposition \ref{prop:deg2}.
	And $\mu = \sum_{y \in \level_k(x)} \mu_y = (1 - o(1)) a.$
	Let $c = \frac {\mu} {\delta n^{1-\maxdeg}} \ge 2 - o(1)$.
	By Chernoff bound,
		\begin{align*}
			\pr{\alpha_k \le \delta n^{1-\maxdeg} \mid \Omega_k}
			\le \exp(- \frac {(c - 1)^2 \delta n^{1-\maxdeg}} {4})
			\sim o(n^{-2})
		\end{align*}
\end{proof}

\begin{proposition}\label{lem:pfp2}
	Let $x$ be a fixed vertex.
	Let $0 \le k \lesssim \log n$.
	Denote by $\Omega^*_k$ the set of graphs such that
		\[ \alpha_i \le \delta n^{1-\maxdeg}, \text{ for any } 0 \le i \le k \]
	Then $\pr{\vol{\level_k(x)} > 4\delta n^{1-\maxdeg}, \Omega^*_k}
		\le (k+1) n^{-2} $.
\end{proposition}

\begin{proof}
	When $k = 0$, the claim is proved by Proposition \ref{lem:volIter}.
	When $k \ge 1$, we will repeatedly apply Proposition \ref{lem:volIter} to
	prove the statement.
	For any values of $i$ smaller than or equal to $k$,
	let	$S_i \subset \Omega^*_k$ denote the set of graphs that also satisfy:
		(1) $\vol{\level_j(x)} \le 4\delta n^{1-\maxdeg}$,
			for any $0 \le j \le i-1$;
		(2) $\vol{\level_k(x)} > 4\delta n^{1-\maxdeg}$.
	We show that $\pr{S_i} - \pr{S_{i+1}} \le n^{-2}$ if $0 \le i \le k-1$,
	and	$\pr{S_k} \le n^{-2}$.
	The conclusion follows from the two claims.

	For the first part,
	\begin{align*}
		\pr{S_i} - \pr{S_{i+1}} &=
			\pr{\vol{\level_i(x)} > 4\delta n^{1-\maxdeg}, S_i} \\
			&\le \pr{\Omega_i, \alpha_i \le \delta n^{1-\maxdeg}} \le n^{-2}
	\end{align*}
	The first inequality is because if $G \in S_i$ and $G$ satisfies
	$\vol{\level_i(x)} > 4\delta n^{1-\maxdeg}$,
	then $G \in \Omega_i$.
	Also $\alpha_i \le \delta n^{1-\maxdeg}$ since $G \in S_i \subset \Omega^*_k$.
	The second inequality is because of Proposition \ref{lem:volIter}.
	The other part can be proved similarly and we omit the details.
\end{proof}

The following proposition helps us control the number of landmarks added in the bottom layer of the local ball.

\begin{proposition}\label{lem:skewdeg}
	Let $x$ be a fixed node with weight $\w x \le 2K$.
	Denote by
		\[ S_x = \set{y \in N(x) : d_x \le d_y \text{ and } d_y \le K} \]
	and let $\dhat x = \size {S_x}$.
	Then
		\[ \pr{\dhat x \ge \max(c_1 {\weight x}^{3 - \plexp}, c_2 \log n)}
			\le n^{-3} \]
	where $c_1 = \frac {192Z}{\avgdeg (\plexp - 2)}$
	and $c_2 = 130$.
\end{proposition}

\begin{proof}
	When $\w x \le c_2 \log n /2$,
		\[ \pr{\dhat x \ge c_2 \log n} \le \pr{d_x \ge c_2 \log n} \le o(n^{-3}) \]
	Now suppose that $\w x > c_2 \log n / 2$.
	Consider any vertex $y$ whose weight is at most $\w x / 8$.
	Then
		\begin{align*}
			\pr{d_y \ge d_x} &\le \pr{d_y \ge \w x / 4} + \pr{d_x < \w x / 4} \sim o(n^{-4})
		\end{align*}
	The second inequality is because of Proposition
	\ref{prop:degree}.
	Hence $y$ is not in $S_x$.\\
	Now if $\w y \ge 2K$, then $\pr{d_y \le K} \le \sim o(n^{-4})$.
	Hence $y$ is also not in $S_x$.
	Lastly,	let $X$ denote the set of vertices whose weight is between $[\frac{\w x} 8, 2K]$ and who is connected to $x$.
	We have
		\begin{align*}
			\Exp{X} &= \sum_{y \in V \setminus \set{x}: \w x / 8 \le \w y \le 2K}
				\frac{\w x \w y} {\vol{V}} \\
				&\le \frac{4 \w x} {{\vol{V}}} \max(\frac {8Z}{\plexp-1} n {\w x}^{2 - \plexp}, \sqrt n \log n) \\
				&\le \max(c_1{\w x}^{3 - \plexp}, c_2 \log n) / 3
		\end{align*}
	The first inequality is because of Proposition \ref{prop:deg2}.
	The second inequality is because ${\vol{V}} = \avgdeg n + o(n)$
	by Proposition \ref{prop:deg2},
	and $\w x \le 2K \le 2\sqrt n$.
	From here it is not hard to obtain that
	$\pr{\size X \le \max(c_1{\w x}^{3 - \plexp}, c_2 \log n)} \sim o(n^{-3})$.
\end{proof}

Now we are ready to bound the output size of Algorithm \ref{alg_approx} with the following Proposition.

\begin{proposition}\label{lem:size2}
	In the setting of Theorem \ref{thm_sp2_ub}, we have that the following holds with high probability:
	\begin{itemize}
		\item $\size {F(x)} \lesssim \max\left(n^{\frac{\beta-2}{\beta-1}}, n^{\frac{3-\beta}{4-\beta}}\right) \cdot \log^3 n$ for all $x \in V$;
		\item The algorithm terminates in time $O\left(\max\left(n^{1+\frac{\beta-2}{\beta-1}}, n^{1+\frac{3-\beta}{4-\beta}}\right) \cdot \log^3 n\right)$.
	\end{itemize}
\end{proposition}


\begin{proof}
	We first bound the number of nodes in $H$.
	By Proposition \ref{lem:heavy}, with probability $1 - n^{-1}$
		\[\size H \lesssim n K^{1-\plexp} \lesssim O(n^{\labelSize}) \]
	Secondly, we bound the number of landmarks added before reaching
	the boundary layer.
	For any vertex $x$, with $i = 0,\dots, l(x)-2$,
	$\size{\level_i(x)} \le \alpha_i(x) = O(\delta n^{1-\maxdeg})$.
	Since $l(x) \le O(\log n)$, the total landmarks for these layers are
	at most $O(n^{1-\maxdeg} \log^3 n)$. The rest of the proof will bound
	the number of landmarks on the boundary layer with depth $l(x)-1$.\\
	Denote by
		\[ \pi_k(x) = \sum_{y \in \level_k(x)} \dhat y \istrue{d_y \le K}
			\quad\text{ for } x \in V, 0 \le k \le n-1 \]
	Hence $\pi_{l(x)-1}(x)$ gives the number of landmarks added
	on the boundary	layer.\\
	Set $c_3 = \frac {3Z} {\abs{2\plexp - 5}} \max(\xmin^{5-2\plexp}, 1)$,
	$\psi = 12c_3 \delta n^{1-\labelExpo}$,
	and $\Delta = \max(c_1 \psi, c_2 \delta n^{1-\maxdeg} \log n)$,
	where $c_1$ and $c_2$ are defined in Proposition \ref{lem:skewdeg}.
	We show that $\pi_{l(x)-1}(x) \le \Delta$ with probability
	$1 - n^{-2}$
	for the rest of the proof ---
	our conclusion follows by taking union bound
	over $x \in V$ and $1 \le l(x) \le O(\log n)$.

	When $l(x) = 1$, $\pi_0(x) = d_x \le K \le \Delta$.
	When $l(x) = k + 1 \ge 2$,
	we know that $G \in \Omega^*_{k-1}$.
	Hence by Proposition \ref{lem:pfp2},
	$\vol{\level_{k-1}(x)} \le 4\delta n^{1-\maxdeg}$ with high probability.
	More concretly,
	\begin{align}
		&\ \pr{l(x) = k+1, \pi_k(x) \ge \Delta} \nonumber \\
		\le&\ (k+1) n^{-2} +
		\Pr[l(x) = k+1,\vol{\level_{k-1}(x)} \le 4\delta n^{1-\maxdeg},
				\pi_k(x) \ge \Delta] \label{eq:i2}
		\end{align}
	Denote by
	\[ w_k = \sum_{y \in \level_k(x)} {\w y}^{3 - \plexp}
		\istrue{\w y \le 2K}. \]
	Conditional on $a = \vol{\level_{k-1}(x)} \le 4\delta n^{1-\maxdeg}$,
	we show that $w_k \le \psi$ with high probability.
	Denote by $\Omega_w$ the set of graphs satisfying $a \le 4\delta n^{1-\maxdeg}$.
	Conditioned on $\Omega_w$,
	$w_k$ is the sum of independent	random variables that are all bounded in
	$[0, (2K)^{3 - \plexp}]$. Hence
		\begin{align*}
			\Exp{w_k} &= \sum_{y \notin \neigh_{k-1}(x)}
				\pr{y \adj \neigh_{k-1}(x)} {\w y}^{3 - \plexp} \istrue{\w y \le 2K} \\
				&\le \frac a {\vol{V}} (\sum_{y \notin \neigh_{k-1}(x)} {\w y}^{4 - \plexp} \istrue{\w y \le 2K})
					&&\text{\hfill (by Proposition \ref{prop:connect})}\\
				&\le \frac a {\vol{V}} (\sum_{y \in V} {\w y}^{4-\plexp} \istrue{\w y \le 2K})	\\
				&\le \frac {a \phi(K) n} {{\vol{V}}} &&\text{ \hfill ({by Proposition \ref{prop:deg2}})} \\
				&\lesssim \frac{a \phi(K)} {\avgdeg}
					&& \text{ (${\vol{V}} = \avgdeg n \pm o(n)$ {by Proposition \ref{prop:deg2}})} \\
				&\le \frac {\psi} 3.
		\end{align*}
	The last line follows by $a \le 4\delta n^{1-\maxdeg}$ 
	and $\phi(K) n^{1-\maxdeg} \le c_2 n^{1-\labelExpo}$.
	Now we apply Chernoff bound on $w_k$,
		\begin{align*}
			\pr{w_k > \psi \mid \Omega_w} \le \exp(-\frac {\psi} {4(2K)^{3-\plexp}})
				\sim o(n^{-2})
		\end{align*}
	because when $2.5 \le \plexp \le 3$,
	\[ \frac {\psi} {K^{3-\plexp}} = \Theta(n^{1 - \maxdeg - \frac {3-\plexp} 2})
		= \Theta(n^{\frac {(\plexp-1)^2 - 2} {2(\plexp-1)}}) \]
	And when $2 < \plexp < 2.5$,
	\[ \frac{\psi} {K^{3-\plexp}} = \Theta(n^{\frac {(3-\plexp)(\plexp-2)} {(4-\plexp)(\plexp-1)}}) \]
	Hence the second part in Equation \eqref{eq:i2} is bounded by
	$o(n^{-2})$ plus
		\begin{align*}
			&\ \Pr[l(x) = k+1, \vol{\level_{k-1}(x)} \le 4\delta n^{1-\maxdeg},w_k \le \psi, \pi_k(x) \ge \Delta] \\
			\le&\ \pr{w_k \le \psi, \alpha_{k-1} \le \delta n^{1-\maxdeg}, \pi_k(x) \ge \Delta} \\
			\le&\ \pr{w_k \le \psi, \size{\level_k(x)} \le \delta n^{1-\maxdeg}, \pi_k(x) \ge \Delta}
		\end{align*}
	In the reminder of the proof we show the above Equation is at most $n^{-2}$.
	Denote by
		\[ \pi'_k(x) = \sum_{y \in \level_k(x)} \dhat y \istrue{\w y \le 2K}\]
	By Proposition \ref{prop:degree},
	$\pr{d_y \le K \mid \w y > 2K} \le \exp(-K / 8) \sim o(n^{-3})$ for any
	$y \in V$.
	Hence $\pi'_k(x) = \pi_k(x)$ with probability at least
	$1 - o(n^{-2})$. Lastly, we have
	\begin{align*}
		\pr{w_k \le \psi, \size{\level_k(x)} \le \delta n^{1-\maxdeg},
			\pi'_k(x) \ge \Delta} \le n^{-2}
	\end{align*}
	Otherwise, there exists a vertex
	$y \in \level_k(x)$ such that $\w y \le 2K$ and
		$\dhat y \ge \max(c_1 {\w y}^{3 -\plexp}, c_2 \log n)$,
	because $\Delta \ge \max(c_1 \psi, c_2 \delta n^{1-\maxdeg} \log n)$.
	This happens with probability at most $n^{-2}$, by taking union bound
	over every vertex with Proposition \ref{lem:skewdeg}.
\end{proof}

\section{Proof of Theorem \ref{thm_lb3}: Lower Bounds for $\plexp > 3$} \label{sec_lb3}

In this section, we present lower bounds for distance labelings on random power law graphs when $\beta > 3$.
We will consider lower bounds for labeling schemes that can estimate all pairs distances up to $K \le \log n / \log r$,\footnote{Note that the average distance of $G$ is $\log n / \log r$ (see e.g. Bollob{\'a}s \cite{B98}).} where $r$ is equal to $\frac{\vols{V}}{\vol{V}}$.
More formally, we say that a labeling scheme is $K$-accurate if for any $x, y \in V$:
\begin{itemize}
	\item[a)] if $\dist(x,y) \le K$, then the labeling scheme returns the exact distance $\dist(x,y)$.
	\item[b)] if $\dist(x,y) > K$, then the labeling scheme returns ``$\dist(x,y) > K$''.
\end{itemize}
Let $d$ be an integer smaller than $K/2$.
\footnote{We assume that $K$ is odd without loss of generality.}

We may assume without loss of generality for every $x$, the label of $x$ stores the distances between $x$ and all vertices in $\neigh_d(x)$.
This is because the lower bound we are aiming at is larger than the size of $\neigh_d(x)$, we can always afford to store them.
From the labels of $x, y$, either we see a non-empty intersection between $\neigh_d(x)$ and $\neigh_d(y)$, which determines their distance; or the two sets are disjoint, in which case we are certain that $\dist(x, y) \geq 2d+1$.
In a random graph, the event that $\dist(x,y) > 2d+1$, conditioned on $\dist(x, y)\geq 2d+1$ and
$\neigh_d(x)$ and $\neigh_d(y)$ are disjoint, happens with probability
\begin{align*}
	\Theta\Bracket{{\frac {\vol{\level_d(x)} \cdot \vol{\level_d(y)}} {\vol{V}}}},
\end{align*}
by Proposition \ref{prop:connect}, assuming that $\vol{\level_d(x)}\vol{\level_d(y)} \le o(\vol{V})$.
Note that this probability gives us a lower bound on the entropy of the event $\mathbf{1}_{\dist(x, y)>2d+1}$.
Since the labels of $x$ and $y$ determine their distance, if we can find a large number of pairwise independent pairs $(x, y)$ such that the entropy of $\mathbf{1}_{\dist(x,y) > 2d+1}$ is large (e.g. $1 / \polylog n$ suffices), then we obtain a lower bound on the total labeling size.

Our discussion so far suggests the following three step proof plan.
\begin{itemize}
	\item[a)] Pick a parameter $d$ and a maximal set of vertices $S$, such that by ``growing'' the local neighborhood of $S$ up to $d$, $\neigh_d(x), \neigh_d(y)$ are disjoint/independent and $\level_d(x), \level_d(x)$ have large volume, for certain pairs of $x, y \in S$.
	\item[b)] Use the labels of $S$ to infer whether there are edges between $\level_d(x)$ and $\level_d(y)$, for certain pairs of $x, y$.
	Obtain a lower bound on the total label length of $S$ via entropic arguments.
	\item[c)] Partition the graph into disjoint groups of size $\abs{S}$.
	Apply step b) for each group.
\end{itemize}

Clearly, given any two vertices, their neighborhood growth are correlated with each other.
However, one would expect that the correlation is small, so long as the volume of the neighborhood has not reached $\Theta(\sqrt n)$.
To leverage this observation, We describe an iterative process to grow the neighborhood of $S$ up to distance $d$.
For simplicity, we assume that $S = \{x_1,x_2,\ldots\}$ only consists of vertices whose weight are all within $[\avgdeg, 2\avgdeg]$.
The motivation is to find \emph{disjoint} sets $L(x_i)$ for each $x_i$, such that $L(x_i)$ is almost as large as $\level_d(x_i)$, and if $\dist(x_i,x_j)>2d+1$, then there is no edge between $L(x_i)$ and $L(x_j)$.

\bigskip
\noindent{\bf The iterative process:}
Denote by $G_1 = (V_1, E_1)$, where $V_1 = V$ and $E_1 = E$.
For any $i \ge 1$, define $T(x_i)$ to be the set of of vertices in $G_i$ whose distance is at most $d$ from $x_i$.
Define $L(x_i)$ to be the set of vertices in $G_i$ whose distance is equal to $d$ from $x_i$.
More formally,
\begin{align*}
	&T(x_i) := \begin{cases}
		           \set{y : \dist_{G_i}(x_i, y) \le d}, &\mbox{if } x_i \in V_i \\
		           \varnothing, &\mbox{otherwise.}
	\end{cases} \\
	&L(x_i) := \set{y \in T(x_i) : \dist_{G_i}(x_i, y) = d}
\end{align*}
We then define $F_i = F_{i-1} \,\cup\, T(x_i)$ ($F_0 := \varnothing$ by default).
Denote by $G_{i+1}$ to be the induced subgraph of $G_i$ on the remaining vertices $V_{i+1} = V \backslash F_i$.

We note that in the above iterative process, the neighborhood growth of $x_i$ only depends on the degree sequence of $V_i$.
We show that under certain conditions, with high probability, a constant fraction of vertices $x \in S$ satisfy that $\vol{L(x)} \geq \Omega(r^d)$.

\begin{lemma}[Martingale inequality]\label{lem_martingale2}
Let $r  = \frac{\vols{V}}{\vol{V}}$.
Let $d$ be an integer and $S \subseteq V$ be a set of vertices whose weight are all within $[\avgdeg, 2\avgdeg]$ and $\abs{S} \le o(\frac {\vol{V}} {\avgdeg^2 r^{d-1}})$.
Assume that
\begin{itemize}
	\item[i)] $\Pr[\vol{L(x_i)} \ge \avgdeg \cdot r^d \mid \vol{F_{i-1}} \le \abs{S} \avgdeg \cdot r^d \log n, x_i \in V_i] \ge \Omega(1)$, for all $1 \le i \le \abs{S}$;
	\item[ii)] $\E[\vol{N_d(x)}] \lesssim \avgdeg\cdot r^d$, for all $x \in S$;
	\item[iii)] $\Pr[\dist(x,y) \le d] \lesssim \frac{\avgdeg^2 \cdot r^{d-1}} {\vol{V}}$, for all $x, y \in S$.
\end{itemize}
Then with high probability, at least $c_1 \abs{S}$ vertices $x \in S$ satisfy that $\vol{L(x)} \ge \Omega(\avgdeg r^{d})$, for a certain fixed constant $c_1$.
\end{lemma}

\begin{proof}
	Consider the following random variable, for any $1 \le i \le \abs{S}$.
	\[ X_i := \begin{cases}
		          1 & \textrm{if } x_i \notin V_i, \textrm{ or } \vol{F_{i-1}} > \abs{S} \avgdeg \cdot r^d \log n,
		          \textrm{or } \vol{L(x_i)} \geq \Omega  (r^{d}) \\
		          0 & \textrm{otherwise.}
	\end{cases} \]
	We have $\Pr[X_i=1\mid X_1,\ldots,X_{i-1}]\geq \Omega(1)$ by Assumption i).
	Thus by Azuma-Hoeffding inequality,
	$\sum_{i=1}^{\abs S} X_i \geq \Omega(\abs{S})$ with high probability.
	We will show below that the contributions to $\sum_{i=1}^{\abs S} X_i$ from the first two predicates is $o(\abs{S})$.
	Hence by taking union bound, we obtain the desired conclusion.

	First, we show that the number of $x_i$ such that $x_i \notin V_i$ is $o(|S|)$ with high probability.
	Note that $x_i \notin V_i$ implies that there exists some vertex $j<i$ such that $\dist(x_i, x_j) \le d$.
	On the other hand, for any two vertices $x, y \in S$,
	$\Pr[\dist(x, y) \le d] \le O(\avgdeg^2 \cdot r^{d-1} / {\vol{V}})$, by Assumption iii).
	Hence, the expected number of vertex pairs in $S$ whose distance is at most $d$,	is $O(\abs{S}^2 \avgdeg^2 \cdot r^{d-1} / {\vol{V}}) \le o(\abs{S})$, by the assumption on the size of $S$.
	By Markov's inequality, with high probability only $o(\abs{S})$	vertex pairs have distance at most $d$ in $S$.
	Hence there exists at most $o(\abs{S})$ $i$'s such that $x_i \notin V_i$.

	Secondly, for all $1 \le i \le \abs{S}$, $\vol{F_i} \le \abs{S} \avgdeg \cdot r^d \log n$  with high probability.
	This is because the set of vertices $T_i$ is a subset of $\neigh_{d}(x_i)$, the set of vertices within distance $d$ to $x_i$ on $G$.
	Thus, by Assumption ii), we have
	\[ \E[\vol{T_i}]\leq \E[\vol{\neigh_d(x_i)}]\leq O(\avgdeg r^d). \]
	And the expected volume of ${F}_{i}$ is at most
	\[ O(i\cdot \avgdeg \cdot r^{d}) \lesssim \abs{S} \avgdeg \cdot r^d, \]
	Hence by Markov's inequality, the probability that $\vol{N_d(S)} > \abs{S} \avgdeg \cdot r^d \log n$ is at most $\log^{-1} n$.
	This proves the lemma.
\end{proof}

We first introduce the following proposition for growing the neighborhood of vertices.

\begin{proposition}[Iterative neighborhood growth]\label{lem_martingale1}
	Let $c = (3+1/\gamma)\log_r\log n$ and $d = K/2 - c$.
	Let $S$ be a set of $n / r^{K/2}$ vertices whose weight are all within $[\avgdeg, 2\avgdeg]$.
	With high probability, at least $c_1 \abs{S}$ vertices in $S$ satisfy that $\vol{L(x)} \ge \Omega(r^d)$.
\end{proposition}

\begin{proof}
	It's easy to verify that $\abs{S} \cdot {\avgdeg^2 r^{d-1}} \le n / r^c \le o(n)$.
	It suffices to verify the assumptions required in Lemma \ref{lem_martingale2}.
	Note that Assumption ii) and iii) simply follows from Proposition \ref{lem_expand}.
	Hence it suffices to verify Assumption i).
	Note that the subgraph $G_i$ can be viewed as a random graph sampled from Chung-Lu model over $V_i$.
	By setting
	\[	p_y^{(i)} = p_y \cdot \left(1 - \frac{\vol{F_{i-1}}} {\vol{V}} \right),
	~\forall\, y \in V_i, \]
	we have that $\forall\, y, z \in V_i$
	\[
		\Pr[y\sim z] = \frac{p_y \cdot p_z}{\vol{V}}
		= \frac {p_y^{(i)} \cdot p_z^{(i)}} {\Bracket{1 - \frac{\vol{F_{i-1}}}{\vol{V}}} \cdot \vol{V_i}}
		= \frac{p_y^{(i)} \cdot p_z^{(i)}} {\sum_{x \in V_i} p_x^{(i)}}.
	\]
	Hence we see that $G_i$ is equivalent to a random graph drawn from degree sequence $\seqdeg^i$.
	Denote by $r_i := \frac{\vols{V_i}} {\vol{V}}$ the growth rate on $G_i$.
	When $\vol{F_{i-1}} \leq n / \log^{2+1/\gamma} n$, by H\"{o}lder's inequality,
	\[ \vols{F_{i-1}}\leq \vol{F_{i-1}}^{\frac{\gamma}{1+\gamma}} \cdot O(n^{\frac{1}{1+\gamma}}) \leq o(n/\log n).\]
	by straightforward calculation.
	Hence $r_i$ is a constant strictly greater than $1$.
	By Proposition~\ref{lem_expand}, with constant probability $\vol{L_i} \geq \Omega(r_i^ {d})\geq \Omega(r^{d})$, because
	\[ \Bracket{\frac{\vol{V}}{\vol{V_i}}}^d \lesssim (1 + \log^{-2-\frac 1 {\gamma}} n)^{O(\log n)} \lesssim 1 + \log^{-1-\frac 1 {\gamma}} n. \]
	Since the vertices at distance $d$ from $x_i$ in $G_i$ is exactly $L_i$,	we have verified that Assumption i) is correct.
\end{proof}

Now we are ready to prove Theorem \ref{thm_lb3}.
\begin{proof}[Proof of Theorem \ref{thm_lb3}]
	Set $K = \ceil{\log n / \log (np)}+1$, i.e. slightly larger than the average distance of $G$.
	We know that there are $\Theta(n)$ vertices whose weights are between $[\avgdeg, 2\avgdeg]$, by an averaging argument.
	Divide them into groups of size $n / r^{K/2}$.
	Clearly, there are $\Theta(r^{K/2}) = \Theta(\sqrt n)$ disjoint groups.
	Denote by $c_2$ a small fixed value (e.g. $1/\log\log n$ suffices).
	We will argue that for each group $S$,
	\begin{align}\label{eq_sp3_group}
		\Pr[\text{The total label length of $S$} \le c_2 \cdot r^{-2c} n] \le o(1).
	\end{align}
	Hence by Markov's inequality, except for $o(n / r^{K/2})$ groups, all the other groups will have label size at least $\tilde{\Omega}(n)$.
	For the rest of the proof, we focus on an individual group $S$.

	Given the labels of $S$, we can recover the pairwise distances less than $K$ for all vertex pairs in $S$.
	Let $\dist_S: S\cdot S \rightarrow \mathbb{N}$ denote the distance function restricted to all pairs in $S$.
	Consider the following two cases:
	\begin{enumerate}
		\item[a)]	$\exists c_1^2 \cdot \abs{S}^2/4$ pairs $(x_i,x_j)$ such that $\dist_S(x_i,x_j)\leq 2d+1$.
		By Lemma~\ref{lem_expand}, we know that $\Pr[\dist(x_i, x_j) \le 2d+1] = O(r^{2d} / n)$, for any $x_i, x_j \in S$.
		Hence the expected number of pairs with distance at most $2d+1$ in $S$,
		is at most $O(\abs{S}^2 \cdot r^{2d}/n) \lesssim r^{-2c}$.
		Hence by Markov's inequality, the probability that a random graph induces any such distance function is $o(1)$.

		\item[b)] The number of pairs such that	$\dist_S(x_i,x_j)\leq 2d+1$ is at most
		$c_1^2\cdot \abs{S}^2/4$ in $S$.
		Let \[ A = \set{(x, y) \in S \cdot S \mid \dist(x, y) > 2d+1, \mbox{ and } \vol{L(x)}, \vol{L(y)} \ge \Omega(r^{d})}. \]
		By Lemma \ref{lem_martingale2}, the size of $A$ is at least
		\[ \binom{c_1 \abs{S}}{2} - c_1^2 \abs{S}^2 /4 \ge c_1^2 \abs{S}^2 / 5. \]
		For any $(x, y) \in A$, $L(x)$ and $L(y)$ are clearly disjoint.
		Conditional on $\set{T(x)}$ for all $x \in S$,
		the probability of the existences of edges between $L_i$ and $L_j$
		are unaffected.
		\begin{align*}
			&\Pr\left[\dist_S(x, y) > 2d+1,~\forall (x, y) \in A \mid \{T_i\}_{i=1}^{\abs S} \right]\\
			&\leq \prod_{(x, y) \in A} \Pr\left[L(x) \not\sim L(y) \mid L(x) \cap L(y) = \varnothing,
			\mbox{ and } \vol{L(x)}, \vol{L(y)} \ge \Omega(r^{d})\right] \\
			&\leq \prod_{(x, y) \in A} \exp\left(-\frac{\vol{L(x)} \vol{L(y)}} {\vol{V}}\right) \tag{(by Proposition \ref{prop:connect})} \\
			&\leq \exp\left(-\Omega\left(\frac {r^{2d}}{n}\right)\right)^{c_1^2 \abs{S}^2 /5} \\
			&\leq \exp(-\Omega(r^{-2c} n)).
		\end{align*}
		Note that the number of labeling of size less than $c_2 \cdot r^{-2c} n$ is at most $2^{c_2\cdot r^{-2c} n}$.
		Therefore by union bound, the probability that the total label size of $\abs S$ is at most $c_2\cdot r^{-2c} n$ is at most:
		\[ 2^{c_2\cdot r^{-2c} n} \cdot \exp(-\Omega(r^{-2c} n)) \le o(1). \]
	\end{enumerate}
	By taking a union bound over cases a) and b), we have shown that Equation \eqref{eq_sp3_group} is true.
	Hence the proof is complete.
\end{proof}

\section{Proof of Theorem \ref{thm_lb_2}: Lower Bounds for $2 < \beta < 3$}\label{sec_pf_lb2}
In this section, we prove the lower bound on distance labelings when $2 < \beta < 3$.

\begin{proof}[Proof of Theorem \ref{thm_lb_2}]
First, one can easily verify that there are $\Theta(n)$ vertices with weight between $[\avgdeg, 2\avgdeg]$ with high probability.
We divide them into  $\Theta(\sqrt n)$ groups of size $\sqrt n$ each.
We will show that for each group, its total labeling length is at least $c n^{3-\plexp}$ with high probability, where $c$ is a fixed constant.
Hence in expectation, except for at most $o(\sqrt n)$ groups, the total labeling length for all other groups is at least $\Omega(n^{3-\plexp})$.
Then by Markov's inequality, among the $\Theta(\sqrt n)$ groups, only a small fraction of $o(\sqrt n)$ groups will  have total labeling lengths at most $c n^{3-\plexp}$ with high probability.
And for the rest of the $\Theta(\sqrt n)$ groups, they have total labeling length at least $c n^{3-\plexp}$.
Hence, we conclude that the total label size is at least $\Omega(n^{3.5-\plexp})$.
For the rest of the proof, we focus on a single group.


Let $\slow$ be a fixed set of $\sqrt n$ vertices, where each vertex has weight between $[\avgdeg, 2\avgdeg]$.
We show that the total labeling length of $\slow$ must be $\Omega(n^{3-\plexp})$ with high probability.
Consider the following procedure of generating a random graph.
It is not hard to verify that it is equivalent to the random power law graph model.
\begin{enumerate}
	\item
		independently for every pair of vertices $x, y$ that are not both in $\shigh$, add an edge between them with probability $\min\{p_xp_y/\vol{V}, 1\}$;
	\item
		independently for every pair of vertices $x, y$ in $\shigh$, add an edge between them with probability $\min\{p_xp_y/\vol{V}, 1\}$.
\end{enumerate}

Let $G_1$ be the (random) graph generated after step 1.
We claim that for $G_1$, with high probability there are $\Omega(n^{3-\plexp})$ pairs $P=\{(x_i, \hat x_i)\}$ of vertices in $\shigh$, and pairs $\{(y_i, \hat y_i)\}$ of vertices in $\slow$ such that:
$x_i$ is $y_i$'s only neighbor in $\shigh$ (same for $\hat x_i$ and $\hat y_i$),
and $d_{G_1}(y_i, \hat y_i) \ge 4$.
We observe the following facts.

\begin{enumerate}[I.]
	\item
		{\bf Many vertices $x_i\in\shigh$ connect to at least one vertex in $\slow$. }

		For each $x_i$, the probability that it has at least one neighbor in $\slow$ is at least
		\[
		\begin{aligned}
			1-\prod_{y\in\slow} (1-p_{x_i}p_y/\vol{V})&\geq 1-\prod_{y\in\slow}(1-\Omega(n^{-0.5})) \\
			&=1-(1-\Omega(n^{-0.5}))^{|\slow|} \\
			&\geq 1-e^{-\Omega(1)}\geq \Omega(1).
		\end{aligned}
		\]
		It is not hard to verify that this event is independent for different $x_i$.
		By Chernoff bound, except with $\exp(-n^{\frac{1}{2}(3-\plexp)})$ probability, $\Omega(n^{\frac{1}{2}(3-\plexp)})$ vertices in $\shigh$ have some neighbor in $\slow$.

		For each $x_i\in\shigh$ that has some neighbor in $\slow$, let $y_i$ be one such neighbor of $x_i$. Let this subset $\{x_i\}$ of $\shigh$ be $\shigh'$, and their neighbors $\{y_i\}$ in $\slow$ be $\slow'$.
	\item
		{\bf For most of $y_i\in \slow'$, $x_i$ is its only neighbor in $G_1$.}

		We first show that few vertices in $\slow$ are connected to more than one vertex in $\shigh$. For each $y\in\slow$, the probability it connects to more than one vertex in $\shigh$ is at most
		\[
			\sum_{x_i,x_j\in \shigh} \frac{p_{x_i}p_y}{\vol{V}}\cdot \frac{p_{x_j}p_y}{\vol{V}}\leq \sum_{x_i,x_j\in \shigh}O(1/n)\leq O(n^{2-\plexp}).
		\]
		This event is again independent for different $y$. By Chernoff bound, except with $\exp(-n^{2.5-\plexp})$ probability, only $O(n^{2.5-\plexp})$ vertices in $\slow$ are connected to more than one vertex in $\shigh$. That is, only a negligible fraction of $\slow'$ may have more than one neighbor in $\shigh$.
		We remove them from $\slow'$ --- let $\slow''$ denote this set and
		let $\shigh''$ denote their corresponding vertices in $\shigh'$.
		There are $\Omega(n^{\frac{1}{2}(3-\plexp)})-O(n^{2.5-\plexp})=\Omega(n^{\frac{1}{2}(3-\plexp)})$ vertices in $\slow''$
		and there is only one neighbor in $\shigh$ for every $v \in \slow''$.

		On the other hand, for each $x\in\slow$, independent of its edges to $\shigh$, the probability that $x$ connects to no other vertex is at least
		\[
		\begin{aligned}
			\prod_y (1-p_yp_x/\vol{V})&\geq \prod_y e^{-O(\avgdeg\cdot p_y/\vol{V})} && \textrm{(by $p_y p_x/\vol{V}=o(1)$)} \\
			&=e^{-O(\avgdeg\cdot \sum_y p_y /\vol{V})} \\
			&\geq e^{-O(1)}=\Omega(1).
		\end{aligned}
		\]
By Chernoff bound, except with $\exp(-n^{\frac{1}{2}(3-\plexp)})$ probability,
$\Omega(n^{\frac{1}{2}(3-\plexp)})$ vertices in $\slow''$ have degree exactly one:
they only connect to one vertex, which is in $\shigh$.
Moreover, by construction, they connect to different vertices in $\shigh$. Let this subset of $\slow''$ be $\slow'''$, and their neighbors be $\shigh'''$.
\end{enumerate}

We set $P$ to be all pairs in $\shigh'''$.
Then $|P|=\Omega(n^{3-\plexp})$.
Based on Fact I and II, for each pair $(x_i, \hat x_i)$ in $P$, we can find a pair $(y_i, \hat y_i) \subseteq \slow''' \times \slow'''$, such that $x_i$ is $y_i$'s only neighbor in $\shigh'''$ (same for $\hat x_i$ and $\hat y_i$).
By construction, the distance between $y_i$ and $\hat y_i$ is at least 4 in $G_1$.
Let $Q$

Now we are ready to finish the proof.
Observe that the actual distance between $y_i$ and $\hat y_i$ tells us whether there is an edge between $x_i$ and $\hat x_i$.
Specifically, $y_i$ and $\hat y_i$ have distance three if and only if $x_i, \hat x_i$ are connected by an edge.
Thus, knowing the pairwise distances of $\slow$ would reveal whether there is an edge between every vertex pair in $P$.
By the definition of distance labeling scheme,
for all $2^{|P|}$ different edge configurations of the vertex pairs in $P$,
we must use different labelings for vertices in $\slow$.
By our definition of $P$, for every pair of vertices,
there is an edge between them with probability within $[1/9, 4/9]$,
and the event is independent of whether there is an edge between any other pairs of vertices.
Therefore, each edge configuration of $P$ appears with probability at most
\[
	(1-1/9)^{|P|}=2^{-c\cdot n^{3-\plexp}},
\] for some constant $c$. The probability that the total label size of $\slow$ is no more than $0.5c\cdot n^{3-\plexp}$ is at most
\[
	2^{0.5c\cdot n^{3-\plexp}}\cdot 2^{-c\cdot n^{3-\plexp}}=2^{-0.5c\cdot n^{3-\plexp}}.
\]

Except with $\exp(-n^{3-\plexp})$ probability, the total labeling length of $\slow$ is at least $\Omega(n^{3-\plexp})$.
	Hence the proof is complete.
\end{proof}

\section{Random Graph Toolbox}\label{sec_tools}

This section summarizes the tools from random graph theory that we use during the proofs.
Most of the technical components are either implicitly or explicitly stated in the literature \cite{CL06}.
We include this section for the completeness of the paper.
Readers who are familiar with random graphs can skip this section.

We first present a proposition, which helps us bound the probability that two sets of vertices are connected by an edge.
\begin{proposition}\label{prop:connect}
	Let $G = (V, E)$ be a random graph with weight sequence $\mathbf{p}$.
	For any two disjoint set of vertices $S$ and $T$,
	\begin{align*}
		1 - \exp\left(- \frac{\vol S \vol T} {\vol{V}}\right) \le \pr{S \sim T} \le \frac {\vol S \vol T} {\vol{V}}.
	\end{align*}
	In particular, when $\vol{S}\vol{T} \le o(\vol{V})$, we have that $\pr{S \sim T} = \Theta\Bracket{\frac{\vol S \vol T} {\vol V}}$.
\end{proposition}

\begin{proof}
	The proof is derived from the following calculations.
	\begin{align*}
		\pr{S \sim T} &= 1 - \prod_{x \in S} \prod_{y \in T} \Bracket{1 - \min(\frac{\w x \w y} {\vol{V}}, 1)}
		\le 1 - \Bracket{1 - \sum_{x \in S}\sum_{y \in T} \min(\frac{\w x \w y} {\vol{V}}, 1)} \\
		&\le 1 - \Bracket{1 - \sum_{x \in S}\sum_{y \in T} \frac{\w x \w y} {\vol{V}}}
		= \frac{\wb S \wb T} {\vol{V}}. \\
		\pr{S \nsim T} &= \prod_{x \in S} \prod_{y \in T} \Bracket{1 - \min(\frac{\w x \w y} {\vol{V}}, 1)}
		\le \exp\left(- \sum_{x \in S} \sum_{y \in T} \min(\frac{\w x \w y} {\vol{V}}, 1) \right) \\
		&\le \exp\left(- \frac {\wb S \wb T} {\vol{V}}\right).
	\end{align*}
\end{proof}

Next, we state a proposition which characterizes the probability that a vertex's actual degree deviates from its weight.

\begin{proposition}\label{prop:degree}
	Let $G = (V, E)$ be a random power law graph.
	Let $x$ be a fixed vertex with weight $\w x$ and degree $d_x$ in $G$.
	Then
	\begin{enumerate}
		\item If $c \ge 3$, then
			\[ \pr{d_x \ge c \w x} \le \exp(-\frac{(c-1)\w x} 2) \]
		\item If $0 < c < 1$, then
			\[ \pr{d_x \le c \w x} \le \exp(-\frac{(1-c)^2 \w x} 8) \]
	\end{enumerate}
\end{proposition}

\begin{proof}
	Let $\mu = \Exp{d_x}$. First,
	\begin{align*}
		\mu &= \sum_{y \in V \setminus \set{x}} \min(\frac {\w x \w y} {\vol{V}}, 1) \\
			&\le \sum_{y \in V} \frac {\w x \w y} {\vol{V}} = \w x
	\end{align*}
	By Chernoff bound, for any $c \ge 3$,
	\begin{align*}
		\pr{d_x \ge c \w x} &\le \exp(-\frac {c \w x - \mu} 2) \\
			&\le \exp(-\frac {(c-1) \w x} 2)
	\end{align*}
	since $c \w x - \mu \ge 2\mu$.\\
	On the other hand,
	let $t = \frac {\avgdeg} {2\maxdegWt} n^{1 - \maxdeg}$,
	then for any $y \in V$ where $\w y \le t$,
	we know that $\w x \w y \le {\vol{V}}$
	by Proposition \ref{prop:deg2}.
	Hence
	\[ \mu \ge \w x(1 - \frac{\w x} {\vol{V}} -
			\frac {\sum_{y \in V} \w y \istrue{\w y \ge t}} {\vol{V}}) \]
	By Proposition \ref{prop:deg2},
		\[ \sum_{y \in V} \w y \istrue{\w y \ge t} \sim o(n) \]
	Since $\w x \sim o(n)$ and ${\vol{V}} = \avgdeg n + o(n)$,
	we conclude that $\mu = \w x(1 - o(1))$.
	By Chernoff bound, for any $0 < c < 1$,
		\begin{align*}
			\pr{d_x \le c \w x} &\le \exp(- \frac{(c \w x - \mu)^2} {4\mu})
			\le \exp(- \frac{\w x (1-c)^2} 8)
		\end{align*}
	for large enough $n$.
\end{proof}

The next proposition helps us characterize the number of vertices whose degree is at least $K$, for a certain value $K$.

\begin{proposition}\label{lem:heavy}
	Let $G = (V, E)$	be a random power law graph.
	Let $8\log n \le K \le \sqrt n$ denote a fixed value
	and $S = \set{x \in V : d_x \ge K}$.
	With probability at least $1 - n^{-1}$,
	$\size S \le 3 \max(\frac {Z 3^{\plexp-1}} {\plexp - 1} n K^{1 - \plexp}, \log n)$.
\end{proposition}

\begin{proof}
	Let $Y_1 = \set {x \in V : \w x \ge \frac K 3}$ and
	$Y_2 = \set {x \in V : \w x < \frac K 3 \text{ and } K \le d_x}$.
	Clearly, $S \subset Y_1 \cup Y_2$.
	We first show that $Y_2$ is empty with probability at least $1 - n^{-1}$.
	Consider a fixed node $x \in V$ with weight $\w x \le K/3$.
	By Proposition \ref{prop:degree},
		\[ \pr{d_x \ge K} \le \exp(-K) \sim o(n^{-2}) \]
	Hence $\pr{Y_2 \neq \emptyset} = o(n^{-1})$ by union bound.

	We then bound the size of $Y_1$.
	The expected value of $Y_1$ is $\frac {Z 3^{\plexp-1}} {\plexp - 1} n K^{1 - \plexp}$.
	Then by Chernoff bound, it's not hard to obtain the desired conclusion (details omitted).
\end{proof}

Lastly, we state the following basic properties of the degree sequence of a random power law graph.

\begin{proposition}\label{prop:deg2}
	Let $f$ denote the probability density	function of a power law
	distribution with mean value $\avgdeg > 1$ and exponent $2 < \plexp \le 3$.
	Let $\seqdeg$ denote $n$ independent samples from $f(\cdot)$.
	Let $\log n \le d \le 2\sqrt n$ be any fixed value
	and let $\maxdegWt$ be a function that
	goes to $0$ when $n$ goes to infinity.
	Then almost surely the following holds:
	\begin{enumerate}[i)]
		\item The maximum weight $\max \seqdeg \ge \maxdegWt n^{\maxdeg}$.
		\item The sum of weights beyond $d$ is
		$\sum_{x \in V} \w x \istrue{\w x \ge d} \sim o(n)$.
		\item The volume of $V$ is ${\vol{V}} = \avgdeg n \pm o(n)$.
		\item Let $\log n < K \le 2\sqrt n$ be a fixed value. Set
		\begin{align*}
			c(K) = \begin{cases}
				       \frac {3 Z \xmin^{5 - 2\plexp}} {2\plexp - 5} &
				       \text{if } 2.5 \le \plexp \le 3 \\
				       \frac {3 Z} {5 - 2\plexp} K^{5 - 2\plexp} &
				       \text{if } 2 < \plexp < 2.5
			\end{cases}
		\end{align*}
		Then
		\[ \sum_{x \in V} {\weight x}^{4 - \plexp} \istrue{\weight x \le K}
		\le c(K) n. \]
		\item Let $c > 1$ denote a fixed constant value. For any vertex $x \in V$,
		\begin{align*}
			\sum_{y \in V} \w y \istrue{\frac {\w y} c \le \weight x \le 2\sqrt n}
			\le
			6 \max(\frac {c^{\plexp - 2} Z} {\plexp - 2} n {\w u}^{2 - \plexp},
			\sqrt n \log n).
		\end{align*}
	\end{enumerate}
\end{proposition}

The proof is via standard concentration inequality (details omitted).

\subsection{Growth Lemma for $\beta > 3$}
\label{sec:growth3}

In this subsection, we analyze and bound the growth rate of the neighborhood sizes for random graphs.
We consider the case of $\beta > 3$, when the variance of the degree distribution is bounded.

\begin{proposition}[Growth rates for $\plexp > 3$]\label{lem_expand}
Let $G = (V, E)$ be a random graph with weight sequence $\seqdeg$ satisfying the following properties:
\begin{itemize}
	\item $\vol{V}=(1+o(1))\avgdeg\cdot n$ for some constant $\avgdeg$;
	\item $\vols{V}=(1+o(1))\secdeg\cdot n$ for some constant $\secdeg$;
	\item $\volg{V}=\tau\cdot n$ for some positive constant $\gamma<1/2$ and $\tau$, where $\volg{S}:=\sum_{x\in S} p_x^{2+\gamma}$;
	\item	The growth rate $r=\frac{\vols{V}}{\vol{V}}$ is bounded away from $1$ ($\avgdeg>\secdeg$).
\end{itemize}
Then for any vertex $x$ with a constant weight, the set of vertices $\level_k(x)$ at distance exactly $k$ from $x$ satisfy that
\begin{enumerate}
	\item
	$\E[\vol{\level_k(x)}]=O\left(r^{k}\right)$ for every $k\leq \log_r n$;
	\item
	$\Pr[\vol{\level_k(x)}\geq \Omega\left(r^{k}\right)]\geq \Omega(1)$ for every positive integer $k \leq \frac{1}{2}\log_r n$.
\end{enumerate}
As a corollary, we have that $\Pr[\dist(x, y) \le k+1] \le O(r^{k} / n)$ for every $k \le \frac 1 2 \log_r n$,
where $y$ is any vertex with constant weight.
\end{proposition}

We first show an upper bound for the expected volume for each level $\level_k(x)$.

\begin{proof}[Proof of Part 1.]
	Let us first fix $\level_k(x)$, and consider the set $\level_{k+1}(x)$. For a vertex $y$, the probability that it is in $\level_{k+1}(x)$ is at most
\[
	p_y\cdot \frac{\vol{\level_k(x)}}{\vol{V}}
\]
by Proposition~\ref{prop:connect}. Thus, the expected volume of $\level_{k+1}(x)$ conditioned on $\vol{\level_k(x)}$ is at most
\begin{align*}
	&\ \E[\vol{\level_{k+1}(x)}\mid\vol{\level_k(x)}] \\
	\leq&\ \sum_{y\not\in \neigh_k(x)}p_y^2\cdot \frac{\vol{\level_k(x)}}{\vol{V}} \\
	\leq&\ \vol{\level_k(x)}\cdot\frac{\vols{V}}{\vol{V}} \\
	=&\ \vol{\Gamma_k(x)}\cdot r.
\end{align*}
On the other hand, $\vol{\level_0(x)}=O(1)$. Thus, we have $\E[\vol{\level_k(x)}]=O(r^k)$.
\end{proof}

Next we present the proof for part 2.
\begin{proof}[Proof of Part 2]
The proof is split into three steps.
\paragraph*{Two fixed constant-weight vertices are close with very low probability}
Fix two vertices $x,y \in S$.
By Item 1, \[
\E[\vol{\neigh_{k}(x)}]\leq O(r^k).
\]

However, for each $i$, the probability that $y$ is at distance $i$ from $x$ conditioned on $\neigh_{i-1}(x)$ is at most
\[
	\Pr[y\in \level_{i}(x)\mid \neigh_{i-1}(x)]\leq p_y\cdot \frac{\vol{\level_{i-1}(x)}}{\vol{V}}.
\]
The probability $y$ is within distance $k+1$ from $x$ is at most
\[
\begin{aligned}
	\Pr[y\in\neigh_{k+1}(x)]&\leq\sum_{i=1}^{k+1}\Pr[y\in\level_{i}(x)] \\
	&\leq \frac{p_y}{\vol{V}}\cdot\sum_{i=1}^{k+1}\E[\vol{\level_{i-1}(x)}] \\
	&=p_y\cdot\frac{\neigh_{k}(x)}{\vol{V}} \\
	&\leq O(r^k).
\end{aligned}
\]

\paragraph*{With large probability, $\level_{k+1}(x)$ has volume not much smaller than $\level_k(x)\cdot\frac{\vols{V\setminus\neigh_k(x)}}{\vol{V}}$}
Conditioned on $\level_k(x)$, the probability that a vertex $y\notin \neigh_k(x)$ is in $\level_{k+1}(x)$ is at least
\[
	1-e^{-p_y\cdot\frac{\vol{\level_k(x)}}{\vol{V}}}
\]
by Proposition~\ref{prop:connect}.

For any $T>0$, we have
\begin{align*}
	\sum_{y:p_y>T}p_y^2 &\leq T^{-\gamma}\cdot\sum_{y:p_y>T}p_y^{2+\gamma}\leq \tau\cdot n\cdot T^{-\gamma}.
\end{align*}
We also have
\[
\begin{aligned}
	\sum_{y:p_y\leq T} p_y^3&\leq T^{1-\gamma}\cdot\sum_{y:p_y\leq T} p_y^{2+\gamma} \leq \tau\cdot n\cdot T^{1-\gamma}.
\end{aligned}
\]

Let us focus on all $y$'s with weight at most $T$.
By that fact that $1-e^{-x}\geq x-x^2/2$ when $x\geq 0$, the expected volume of $\level_{k+1}(x)\cap \{y:p_y\leq T\}$ conditioned on $\neigh_k(x)$ is at least

\begin{align*}
	&\ \sum_{y\not\in\neigh_k(x):p_y\leq T}p_y\left(1-e^{-p_y\cdot\frac{\vol{\level_k(x)}}{\vol{V}}}\right)\\
	\geq&\ \sum_{y\not\in\neigh_k(x):p_y\leq T} p_y^2\cdot\frac{\vol{\level_k(x)}}{\vol{V}}- \frac{1}{2}\sum_{y\not\in\neigh_k(x):p_y\leq T} p_y^3\cdot\left(\frac{\vol{\level_k(x)}}{\vol{V}}\right)^2\\
	\geq&\ \left(\vols{V\setminus\neigh_k(x)}\right)\cdot \frac{\vol{\level_k(x)}}{\vol{V}}- \sum_{y:p_y\geq T} p_y^2\cdot\frac{\vol{\level_k(x)}}{\vol{V}}- \frac{1}{2}\sum_{y:p_y\leq T} p_y^3\cdot\left(\frac{\vol{\level_k(x)}}{\vol{V}}\right)^2\\
	\geq&\ \vol{\level_k(x)}\cdot\frac{\vols{V\setminus\neigh_k(x)}}{\vol{V}} - \tau\cdot n\cdot T^{-\gamma}\cdot \frac{\vol{\level_k(x)}}{\vol{V}}- \tau\cdot n\cdot T^{1-\gamma}\cdot \left(\frac{\vol{\level_k(x)}}{\vol{V}}\right)^2 \\
	\geq&\ \vol{\level_k(x)}\cdot\left(\frac{\vols{V\setminus\neigh_k(x)}}{\vol{V}}-\tau\cdot \left(T^{-\gamma}+T^{1-\gamma}\cdot \frac{\vol{\level_k(x)}}{\vol{V}}\right)\right).
\end{align*}

Note that ``$y\in \level_{k+1}(x)$'' are independent events conditioned on $\neigh_k(x)$ for different $y\notin\neigh_k(x)$.
Now we apply Chernoff Bound to lower bound the probability that the volume of $\level_{k+1}(x)$ is too small.
The above inequality holds for every $T>0$.
In the following, we set $T=\vol{\level_k(x)}^{1/2}$.

When $\vol{\level_k(x)}\leq \vol{V}^{2/3}$,
$T^{-\gamma}\geq T^{1-\gamma}\cdot \frac{\vol{\level_k(x)}}{\vol{V}}$,
the expected volume of $\level_{k+1}(x)$ conditioned on $\neigh_k(x)$ is at least:
\begin{align*}
	&\ \E[\vol{\level_{k+1}(x)\cap \{y:p_y\leq T\}}\mid \neigh_k(x)]\\
	\geq&\ \vol{\level_k(x)}\cdot\left(\frac{\vols{V\setminus\neigh_k(x)}}{\vol{V}}-2\tau\cdot \vol{\level_k(x)}^{-\gamma/2}\right).
\end{align*}
Since each $p_y\leq T=\vol{\level_k(x)}^{1/2}$, by Chernoff bound, we have
\begin{align}
	&\Pr\left[\vol{\level_{k+1}(x)}\leq \vol{\level_k(x)}\cdot\left(\frac{\vols{V\setminus\neigh_k(x)}}{\vol{V}}-\vol{\level_k(x)}^{-\gamma/3}\right)\,\middle|\, \neigh_k(x)\right] \notag\\
	&\leq 2^{-\Theta\left(\vol{\level_k(x)}^{1/2-2\gamma/3}\right)}, \label{eq_level_ineq}
\end{align}
as long as $\vol{\level_k(x)}=O(n^{2/3})$ and $\vol{\level_k(x)}$ sufficiently large.


\paragraph*{With constant probability, $\level_k(x)$ has volume at least $\Omega(r^k)$}
Fix a sufficiently large constant $C$, denote by $\cE_0$ the event that $x$ has a neighborhood of volume at least $C$.
Then it is not hard to verify that for any constant $C$, the probability of $\cE_0$ is at least a constant:
 $$\Pr[\vol{\level_1(x)}\geq C]\geq \Omega_C(1).$$
Moreover, for $i\geq 1$, denote by $\cE_i$ the event that either
\begin{align*}
	\vol{\level_{i+1}(x)}&> \vol{\level_i(x)}\cdot \left(\frac{\vols{V\setminus\neigh_i(x)}}{\vol{V}}-\vol{\level_i(x)}^{-\gamma/3}\right)
\end{align*}
or
\[
	\vol{\level_i(x)}\geq n^{2/3}.
\]
By the argument above, $$\Pr[\overline{\cE_i}\mid \neigh_i(x)]\leq 2^{-\Theta(\vol{\level_i(x)}^{1/2-2\gamma/3})}.$$
We claim that these events have the following properties.
\begin{claim}\label{cl_size}
	When $\cE_i$ occurs for all $0\leq i<k$, we must have either $\vol{\level_k(x)}\geq \Omega(r^k)$ or $\vol{\neigh_k(x)}\geq n^{2/3}$ for sufficiently large $C$.
\end{claim}
\begin{claim}\label{cl_prob}
	All events $\cE_i$'s $(0\leq i<k)$ occur simultaneously with constant probability.
\end{claim}
Before proving the two claims, let us first show that they together imply that $\Pr[\vol{\level_k(x)}\geq \Omega(r^k)]\geq \Omega(1)$.

By Markov's inequality, the first inequality in the lemma statement and $k\leq \frac{1}{2}\log_r n$, we have
\[
\Pr[\vol{\neigh_k(x)}\geq n^{2/3}]\leq O(\sqrt{n}/n^{2/3})=o(1).
\]
Therefore, we have the lower bound
\begin{align*}
	&\ \Pr[\vol{\level_k(x)}\geq \Omega(r^k)]\\
	\geq&\ \Pr[\cE_0,\ldots,\cE_{k-1}]\cdot \Pr[\vol{\level_k(x)}\geq \Omega(r^k)\mid \cE_0,\ldots,\cE_{k-1}] \\
	\geq&\ \Pr[\cE_0,\ldots,\cE_{k-1}]\cdot (1-\Pr[\vol{\neigh_k(x)}\geq n^{2/3}\mid \cE_0,\ldots,\cE_{k-1}]) \\
	\geq&\ \Pr[\cE_0,\ldots,\cE_{k-1}]\cdot (1-\Pr[\vol{\neigh_k(x)}\geq n^{2/3}]/\Pr[\cE_0,\ldots,\cE_{k-1}]) \\
	=&\ \Omega(1).
\end{align*}
This concludes the proof.
\end{proof}

\begin{proof}[Proof of Claim~\ref{cl_size}]
Assume $\cE_i$ occurs for all $0\leq i<k$ and $\vol{\neigh_k(x)}< n^{2/3}$.
The goal is to show that in this case, we must have $\vol{\level_k(x)}\geq \Omega(r^k)$.

In particular, $\vol{\neigh_k(x)}< n^{2/3}$ implies that $\vol{\neigh_i(x)}< n^{2/3}$ and $\vol{\level_i(x)}\leq n^{2/3}$ for every $i\leq k$.
By H\"{o}lder's inequality, we also have
\begin{align*}
\vols{\neigh_i(x)}&\leq \volg{\neigh_i(x)}^{\frac{1}{1+\gamma}}\cdot \vol{\neigh_i(x)}^{\frac{\gamma}{1+\gamma}}\\
&\leq \tau^{\frac{1}{1+\gamma}}\cdot n^{1-\frac{\gamma}{3(1+\gamma)}}.
\end{align*}

Thus, the event $\cE_i$ ($i>0$) implies
\begin{align}
	\vol{\level_{i+1}(x)}&>\vol{\level_i(x)}\cdot \left(r-\tau^{\frac{1}{1+\gamma}}\cdot n^{-\frac{\gamma}{3(1+\gamma)}}-\vol{\level_i(x)}^{-\gamma/3}\right).\label{eq_vol_level}
\end{align}
Let $\hat{r}=r-\tau^{\frac{1}{1+\gamma}}\cdot n^{-\frac{\gamma}{3(1+\gamma)}}-C^{-\gamma/3}$.
For sufficiently large $C$, $\hat{r}>\sqrt{r}>1$.
First we can show $\vol{\level_i(x)}\geq C\cdot \hat{r}^{i-1}\geq C\cdot r^{(i-1)/2}$ for $i\geq 1$ inductively:
\begin{itemize}
	\item
		By the definition of $\cE_0$, $\vol{\level_1(x)}\geq C$;
	\item
		If $\vol{\level_i(x)}\geq C\cdot \hat{r}^{i-1}$, then we have
		\[
		\begin{aligned}
		\vol{\level_{i+1}(x)}&\geq \vol{\level_i(x)}\cdot \left(r-\tau^{\frac{1}{1+\gamma}}\cdot n^{-\frac{\gamma}{3(1+\gamma)}}-\vol{\level_i(x)}^{-\gamma/3}\right) \\
		&\geq \vol{\level_i(x)}\cdot \left(r-\tau^{\frac{1}{1+\gamma}}\cdot n^{-\frac{\gamma}{3(1+\gamma)}}-C^{-\gamma/3}\right) \\
		&\geq \vol{\level_i(x)}\cdot \hat{r}.
		\end{aligned}
		\]
\end{itemize}

Thus, we have $\vol{\level_i(x)}\geq\Omega(\hat{r}^{i})\geq \Omega(r^{i/2})$.
By Equation~\eqref{eq_vol_level} again, we have
\[
\begin{aligned}
	\vol{\level_k(x)}&>\vol{\level_{k-1}(x)}\cdot(r-\tau^{\frac{1}{1+\gamma}}\cdot n^{-\frac{\gamma}{3(1+\gamma)}}-r^{-(k-1)\gamma/6}) \\
	&\geq \vol{\level_1(x)}\cdot \left(\prod_{i=1}^{k-1}(r-\tau^{\frac{1}{1+\gamma}}\cdot n^{-\frac{\gamma}{3(1+\gamma)}}-r^{-i\gamma/6})\right) \\
	&\geq r^{k-1}\cdot C\cdot \left(\prod_{i=1}^{k-1}(1-\tau^{\frac{1}{1+\gamma}}\cdot n^{-\frac{\gamma}{3(1+\gamma)}}\cdot r^{-1}-r^{-i\gamma/6-1})\right) \\
	&= r^k\cdot \alpha_k,
\end{aligned}
\]
where $\alpha_k$ is decreasing as $k$ increases. $\alpha_{\frac{1}{2}\log_r n}$ is lower bounded by a constant $\alpha$. Thus, $\vol{\level_k(x)}\geq \alpha\cdot r^k\geq \Omega(r^k)$.
This proves the claim.
\end{proof}

\begin{proof}[Proof of Claim~\ref{cl_prob}]{
By Lemma~\ref{cl_size}, conditioned on $\cE_0,\ldots\cE_{i-1}$, we have either
\[
	\vol{\level_i(x)}\geq \alpha\cdot r^i
\]
or
\[
	\vol{\neigh_i(x)}\geq n^{2/3}.
\]

Thus, by Equation~\eqref{eq_level_ineq}, we have
\[
	\Pr[\overline{\cE_i}\mid \cE_0,\ldots,\cE_{i-1}]\leq 2^{-\Theta(r^{\Omega(i)})}.
\]

Since $\Pr[\cE_0]=\Omega(1)$, we may lower bound the probability that all events happen simultaneously
\[
\begin{aligned}
	\Pr[\cE_0,\ldots,\cE_{k-1}]&\geq \Omega\left(\prod_{i=0}^{k-1}\left(1-2^{-\Theta(r^{\Omega(i)})}\right)\right) \\
	&\geq \Omega(1).
\end{aligned}
\]
}\end{proof}

\end{document}